%% file: document.tex
\documentclass[envcountsame,runningheads]{llncs}
\usepackage{marvosym}
\usepackage{amssymb, amsmath}
\usepackage{mathtools}
\usepackage[all]{xy}
\usepackage{proof}%,bussproofs}
\usepackage{stackengine}
\usepackage{stmaryrd}
\stackMath
\usepackage{tikz}
\usetikzlibrary{arrows,shapes,backgrounds,positioning}%,shapes.callouts,shapes.geometric}
\tikzset{
	itria/.style={
		draw,dashed, %shape border uses incircle,
		isosceles triangle,shape border rotate=90,yshift=-7ex}
}
\newcommand*{\sys}[1]{\ensuremath{\mathsf{#1}}}
\renewcommand*{\l}{\mathsf{L}}
\newcommand*{\f}{\mathsf{f}}

\newcommand*{\bi}{\sys{BI}}
\newcommand*{\lbi}{\l\bi}
\newcommand*{\fbi}{{\f\bi}}

\newcommand*{\rn}[1]  {\ensuremath{\mathsf{#1}}}
\newcommand*{\cut}{\rn{cut}}
\newcommand*{\weak}{\rn{W}}
\newcommand*{\exch}{\rn{E}}
\newcommand*{\cont}{\rn{C}}
\newcommand*{\rrn}[2][]  {\rn{#2_{R#1}}}
\newcommand*{\lrn}[2][]  {\rn{#2_{L#1}}}
\newcommand{\B}{\mathbb{B}}

\newcommand{\Formulas}{\mathbb{F}}
\newcommand{\nests}{\mathbb{B}\scriptstyle{/\equiv}\textstyle}
\renewcommand{\i}{\top^*}
\newcommand\wand{\mathrel{-\mkern-6mu*}}

\newcommand\upshift{{\uparrow}}
\newcommand\downshift{{\downarrow}}
\newcommand{\nestify}{\eta}
\newcommand{\bunch}{\beta}

\newcommand{\foc}[1]{\textcolor{blue}{\langle #1 \rangle}}
\newcommand{\floor}[1]{\lfloor #1 \rfloor}

\newcommand\D{\mathcal{D}}

\renewcommand{\vec}[1]{\overrightarrow{#1}}
\renewcommand{\emptyset}{\varnothing}
\renewcommand{\phi}{\varphi}
\newcommand{\seq}{\Rightarrow}
\title{Focused Proof-search in the Logic of Bunched Implications\thanks{This work has been partially supported by the UK’s EPSRC through research grant EP/S013008/1.}}
\author{Alexander Gheorghiu(\Letter) \and
	Sonia Marin(\Letter)
}
\authorrunning{A. Gheorghiu \and S. Marin.}
\institute{University College London, London, United Kingdom \\
	\email{\{alexander.gheorghiu.19, s.marin\}@ucl.ac.uk}\\
}

\input{figures}

\begin{document}

	\maketitle              % typeset the header of the contribution

\begin{abstract}
			The logic of Bunched Implications (BI) freely combines additive and multiplicative connectives, including implications; however, despite its well-studied proof theory, proof-search in BI has always been a difficult problem. The focusing principle is a restriction of the proof-search space that can capture various goal-directed proof-search procedures. In this paper we show that focused proof-search is complete for BI by first reformulating the traditional bunched sequent calculus using the simpler data-structure of nested sequents, following with a polarised and focused variant that we show is sound and complete via a cut-elimination argument. This establishes an operational semantics for focused proof-search in the logic of Bunched Implications.
	
\end{abstract}

\section{Introduction}
\input{introduction}
	
\section{Re-presentations of BI} \label{sec:BI}
\input{background}

\section{A Focused System} \label{sec:foc} 
\input{focusedsystem}

\section{Conclusion}\label{sec:conclusion}
\input{conclusion}

\bibliographystyle{splncs04}
\bibliography{bib}

%%%%% To display Open Access text and logo, Please add below text and copy the cc_by_4-0.eps in the manuscript package %%%

\vfill

{\small\medskip\noindent{\bf Open Access} This chapter is licensed under the terms of the Creative Commons\break Attribution 4.0 International License (\url{http://creativecommons.org/licenses/by/4.0/}), which permits use, sharing, adaptation, distribution and reproduction in any medium or format, as long as you give appropriate credit to the original author(s) and the source, provide a link to the Creative Commons license and indicate if changes were made.}

{\small \spaceskip .28em plus .1em minus .1em The images or other third party material in this chapter are included in the chapter's Creative Commons license, unless indicated otherwise in a credit line to the material.~If material is not included in the chapter's Creative Commons license and your intended\break use is not permitted by statutory regulation or exceeds the permitted use, you will need to obtain permission directly from the copyright holder.}

\medskip\noindent\includegraphics{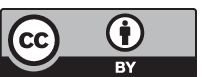}

%
%	
%	\newpage
%\appendix
%
%\newpage
%\section{Bad Cuts}\label{app:badcuts}
%	\input{app-badcuts}
%
%\newpage
%\section{Good Cuts}\label{app:goodcuts}
%	\input{app-goodcuts}
%
%\newpage
%\section{Simulation}\label{app:simulation}
%	\input{app-simulation}
\end{document}

%% file: figures.tex
\newcommand{\LBI}{
		\begin{figure}[t] 
			\fbox{
				\begin{minipage}{.95\linewidth}
\centering
$
\infer[\rn{Ax}]{A \seq A}{}
\quad
\infer[\lrn\bot]{\Delta(\bot) \seq \phi}{}
\quad
\infer[\rrn\i]{\emptyset_\times \seq \i}{}
\quad
\infer[\rrn\top]{\emptyset_+ \seq \top}{}
$
\\[1.5ex]
$
\infer[\lrn\wand]{\Delta(\Delta',\Delta'',\phi \wand \psi) \seq \chi}{\Delta' \seq \phi & \Delta(\Delta'',\psi) \seq \chi}
\quad
\infer[\rrn\wand]{\Delta \seq \phi \wand \psi}{\Delta,\phi \seq \psi}
$
\\[1.5ex]
$
\infer[\lrn\ast]{\Delta(\phi * \psi) \seq \chi}{\Delta(\phi,\psi) \seq \chi}
\quad
\infer[\rrn\ast]{\Delta,\Delta' \seq \phi*\psi}{\Delta \seq \phi & \Delta' \seq \psi}
\quad
\infer[\lrn\i]{\Delta(\top^*) \seq \chi}{\Delta(\emptyset_\times)\seq \chi}
$
\\[1.5ex]
$
\infer[\lrn\land]{\Delta(\phi \land \psi) \seq \chi}{\Delta(\phi;\psi) \seq \chi}
\quad
\infer[\rrn\land]{\Delta;\Delta' \seq \phi\land\psi}{\Delta \seq \phi & \Delta' \seq \psi}
\quad
\infer[\lrn\top]{\Delta(\top) \seq \chi}{\Delta(\emptyset_+)\seq \chi}
$
\\[1.5ex]
$
\infer[\lrn\lor]{\Delta(\phi \lor \psi) \seq \chi}{\Delta ( \phi  ) \seq \chi & \Delta ( \psi ) \seq \chi }
\quad
\infer[\rrn\lor_1]{\Delta \seq \phi \lor \psi}{\Delta \seq  \phi}
\quad
\infer[\rrn\lor_2]{\Delta \seq \phi \lor \psi}{\Delta \seq  \psi}
$
\\[1.5ex]
$
\infer[\lrn\to]{\Delta(\Delta';\Delta'';\phi \to \psi) \seq \chi}{\Delta' \seq \phi & \Delta(\Delta'';\psi) \seq \chi}		
\quad
\infer[\rrn\to]{\Delta \seq \phi \to \psi}{\Delta;\phi \seq \psi}			
\quad
\infer[\cont]{\Delta(\Delta') \Rightarrow \chi}{\Delta(\Delta';\Delta') \Rightarrow \chi}
$
\\[1.5ex]
$
\infer[\weak]{\Delta(\Delta';\Delta'') \seq \chi}{\Delta(\Delta') \seq \chi}
\quad
\infer[\exch_{(\Delta \equiv \Delta')}]{\Delta' \seq \chi}{\Delta \seq \chi}
\quad\infer[\cut]{\Delta(\Delta') \Rightarrow \chi}{\Delta' \Rightarrow \phi & \Delta(\phi) \Rightarrow \chi}
$
				\end{minipage}
			}
			\caption{Sequent Calculus $\lbi$}
			\label{fig:lbi}
		\end{figure}
	}

\newcommand{\etalbi}{
		\begin{figure}[t] 
	\fbox{
		\begin{minipage}{0.95\linewidth}
\centering
$
\infer[\rn{Ax}]{\{\Gamma,A\}_+ \seq A}{}
\quad
\infer[\lrn\bot]{\Gamma\{\bot\} \seq \chi}{}
\quad
\infer[\rrn\i]{\emptyset_\times \seq \i}{}
\quad
\infer[\rrn\top]{\Gamma \seq \top}{}
$
\\[1.5ex]
$
\infer[\lrn\wand]{\Gamma\{\Gamma',\Gamma'',\{\Gamma''', \phi \wand \psi \}_+ \}_\times \seq \chi}{\Gamma' \seq \phi & \Gamma\{\Gamma'',\psi \}_\times \seq \chi}
\quad
\infer[\rrn\wand]{\Gamma \seq \phi \wand \psi }{\{\Gamma,\phi\}_\times \seq \psi}
$
\\[1.5ex]
$
\infer[\lrn\ast]{\Gamma\{\phi * \psi \} \seq \chi}{\Gamma\{\{\phi,\psi\}_\times\} \seq \chi}
\quad
\infer[\rrn\ast]{\{\{\Gamma,\Gamma'\}_\times \,,\Gamma''  \}_+ \seq \phi*\psi }{\Gamma\seq \phi & \Gamma' \seq \psi}
\quad
\infer[\lrn\i]{\Gamma \{\i \} \seq \chi}{\Gamma\{\emptyset_\times\} \seq \chi}
$
\\[1.5ex]
$
\infer[\lrn\land]{\Gamma \{ \phi \land \psi \} \seq \chi}{\Gamma \{\{ \phi,\psi\}_+\} \seq \chi}
\quad
\infer[\rrn\land]{\Gamma \seq \phi\land\psi}{\Gamma \seq \phi & \Gamma \seq \psi}
\quad
\infer[\lrn\top]{\Gamma\{\top \} \seq \chi}{\Gamma \{ \emptyset_+ \} \seq \chi}
$
\\[1.5ex]
$
\infer[\lrn\lor]{\Gamma\{\phi \lor \psi \} \seq \chi}{\Gamma \{ \phi  \} \seq \chi & \Gamma \{ \psi \} \seq \chi}
\quad
\infer[\rrn\lor_1]{\Gamma \seq \phi \lor \psi}{\Gamma \seq \phi}
\quad
\infer[\rrn\lor_2]{\Gamma \seq \phi \lor \psi}{\Gamma \seq \psi}
$
\\[1.5ex]
$
\infer[\lrn\to]{\Gamma\{\Gamma',\phi \to \psi \}_+ \seq \chi}{\Gamma' \seq \phi & \Gamma\{\Gamma',\psi \}_+ \seq \chi}
\quad
\infer[\rrn\to]{\Gamma \seq \phi \to \psi }{\{\Gamma,\phi \}_+ \seq \psi}
\quad
\infer[\cont]{\Gamma\{\Gamma'\}_+\seq \chi}{\Gamma\{\Gamma',\Gamma'\}_+ \seq \chi}
$
		\end{minipage}
	}
	\caption{Sequent Calculus $\nestify\lbi$} \label{fig:taulbi}
\end{figure}
}

\newcommand{\focbi}{
\begin{figure}[t]
	\fbox{
		\begin{minipage}{.95\linewidth}
			\centering
			\textbf{Focused} 
			\\[1ex]
			$
			\infer[\rn{Ax^+}]{\{\Gamma, A_+\}_+ \seq  \foc{A_+} }{}
			\quad
			\infer[\rn{Ax^-}]{\{\Gamma, \foc{A_-}\}_+ \seq A_- }{}
			\quad
			\infer[\rrn{\top^+}]{\vec\Gamma \seq \foc{\top^+} }{}
			$
			\\ [1ex]
			$
			\infer[\rn P]{\{\vec\Gamma,\downshift \upshift P\}_+ \seq \foc{P}}{}
			\quad
			\infer[\rn N]{\{\vec\Gamma,\foc{N}\}_+  \seq \upshift \downshift  N}{}
			\quad
			\infer[\rrn{\i}]{\{\vec\Gamma,\emptyset_\times\}_+ \seq  \foc{\i} }{}
			$
			\\ [1ex]
			$
			\infer[\text{$\lrn[i]{\land^- }$}]{\vec{\Gamma}\{ \foc{N_1 \land^- N_2} \}_+ \seq R}{\vec{\Gamma}\{\foc{N_i}\}_+ \seq R}
			\quad
			\infer[\text{$\rrn[i]{\lor }$}]{\vec{\Gamma} \seq  \foc{P_1 \lor P_2} }{\vec{\Gamma}\seq \foc{P_i} }
			\quad
			\infer[\lrn{\top^-}]{\vec\Gamma\{\foc{\top^-}\} \seq R}{\vec\Gamma\{\emptyset_+\} \seq R}
			$
			\\ [1ex]
			$
			\infer[\rrn{\land^+}]{\{ \vec{\Gamma},\vec{\Gamma}' \}_+ \seq \foc{P \land^+ Q}}{\vec{\Gamma} \seq \foc{P} & \vec{\Gamma}' \seq \foc{Q}}
			\quad
			\infer[\lrn\to]{\vec{\Gamma}\{\vec{\Delta},\foc{P \to N} \}_+ \seq R}{
				\vec{\Delta} \seq \foc{P} 
				&
				\vec{\Gamma}\{\vec{\Delta},\foc{N}\}_+ \seq R
				}
			$
			\\ [1ex]
			$
			\infer[\rrn\ast]{\{\{\vec{\Gamma},\vec{\Gamma}'\}_\times, \vec{\Gamma}''\}_+ \seq \foc{P * Q}}{
				\vec{\Gamma}\seq \foc{P} 
				&
				\vec{\Gamma}' \seq  \foc{Q}
			}
			\quad
			\infer[\lrn\wand]{\vec{\Gamma}\{\vec{\Delta},\vec{\Delta}', \{\vec{\Delta}'',\foc{P \wand N} \}_+ \}_\times \seq R}{
				\vec{\Delta} \seq \foc{P}
				&
				\vec{\Gamma}\{\vec{\Delta}',\foc{N}\}_\times \seq  R
				}
			$
			\\ [1.5ex]
			\textbf{Neutral}
			\\ [1.5ex]
			$
			\infer[\rrn\upshift]{\vec{\Gamma} \seq \upshift P}{\vec{\Gamma} \seq \foc{P}}
			\quad
			\infer[\lrn\upshift]{\vec{\Gamma}\{\foc{\upshift P}\} \seq R}{\vec\Gamma\{P\} \seq R}
			\quad
			\infer[\rrn\downshift]{\vec{\Gamma} \seq \foc{\downshift N}}{\vec{\Gamma} \seq N}
			\quad
			\infer[\lrn\downshift]{\vec{\Gamma}\{\downshift N\} \seq R}{\vec{\Gamma}\{ \foc{N}\} \seq R}
			$
			\\ [1ex]
			$
			\infer[\rn C]{\vec{\Gamma}\{\vec{\Delta}\} \seq R}{\vec{\Gamma}\{\{\vec{\Delta},\vec{\Delta}\}_+\} \seq R}
			$
			\\ [1.5ex]
			\textbf{Unfocused}
			\\ [1ex]
			$
			\infer[\rrn{\top^-}]{\Gamma \seq  \top^- }{}
			\quad
			\infer[\lrn{\bot}]{\Gamma\{\bot\} \seq N}{}
			$
			\\ [1ex]
			$
			\infer[\rrn{\land^-}]{\Gamma \seq N \land^- M}{\Gamma \seq N & \Gamma \seq M}
			\quad
			\infer[\lrn\lor]{\Gamma\{P\lor Q\} \seq N}{\Gamma\{P\} \seq N & \Gamma\{Q\} \seq N}
			$
			\\ [1ex]
			$
			\infer[\lrn{\land^+}]{\Gamma\{P\land^+ Q\} \seq  N}{\Gamma\{\{P, Q\}_+\} \seq  N}
			\quad
			\infer[\rrn\to]{\Gamma \seq P \to N}{\{\Gamma, P \}_+ \seq N}
			\quad
			\infer[\lrn{\top^+}]{\Gamma\{\top^+ \} \seq N}{\Gamma \{ \emptyset_+ \} \seq N}
			$
			\\ [1ex]
			$
			\infer[\lrn\ast]{\Gamma\{P* Q\} \seq N}{\Gamma\{\{P,Q\}_\times\} \seq N}
			\quad
			\infer[\rrn\wand]{\Gamma \seq P \wand N}{\{\Gamma, P \}_\times \seq  N}
			\quad
			\infer[\lrn\i]{\Gamma \{\i \} \seq N}{\Gamma\{\emptyset_\times\} \seq N}
			$
		\end{minipage}
	}
	\caption{System $\fbi$} \label{fig:fbi}
\end{figure}	
}

%% file: introduction.tex
The \emph{Logic of Bunched Implications} (BI)~\cite{Hearn99} is well-known for its applications in systems modelling \cite{Pym19}, especially a particular theory (of a variant of BI) called \emph{Separation Logic}~\cite{Reynolds02,Ishtiaq2011} which has found industrial use in program verification. 
%Despite having a well-developed and well-studied proof theory \cite{Pym02}, comparatively little formal study has been made on proof-search in BI. 
%\todo{probably need to tone this one down}
In this work, we study an aspect of proof search in BI, relying on its well-developed and well-studied proof theory~\cite{Pym02}.
We show that a goal-directed proof-search procedure known as \emph{focused proof-search} is complete; that is, if there is a proof then there is a focused one. Focused proofs are both interesting in the abstract, giving insight into the proof theory of the logic, and have (for other logics) been a useful modelling technology in applied settings. For example, focused proof-search forms an operational semantics of the DPLL SAT-solvers~\cite{Farooque13}, logic programming~\cite{Miller91,Andreoli92,Dyckhoff06,Chaudhuri06}, automated theorem provers \cite{Mclaughlin08}, and has been successful in providing a meta-theoretic framework in intuitionistic, substructural, and modal logics~\cite{Marin16,Miller13,Liang09}. 

Syntactically BI combines additive and multiplicative connectives, but unlike related logics such as Linear Logic (LL) \cite{Girard87}, BI takes all the connectives as primitive. Indeed, it arose from a proof-theoretic investigation on the relationship between conjunction and implication. As a result, sequents in BI have a more complicated structure: each implication comes with an associated context-former. Therefore, in BI contexts are not lists, nor multisets, but instead are \emph{bunches}: binary trees whose leaves are formulas and internal nodes context-formers. Additive composition $(\Gamma;\Delta)$ admits the structural rules of weakening and contraction, whereas multiplicative composition $(\Gamma, \Delta)$ denies them. The principal technical challenges when studying proof-search in BI arise from the interaction between the additive and multiplicative fragments. We overcome these challenges by restricting the application of structural rules in the sequent calculus $\lbi$ as well as working with a representation of bunches as nested multisets. 

%\todo{this paragraph is good but disrupts the flow a bit, should probably be reformulated or placed somewhere else}
Throughout we use the term \emph{sequent calculus} in a strict sense; that is, meaning a label-free internal sequent calculus, formed in the case of BI by a context (a bunch) and a consequent (a formula). The term \emph{proof-search} is consistently understood to be read as backward reduction within such a system. Although there is an extensive body of research on systems and procedures for semantics-based calculi in BI \cite{Galmiche2001,Galmiche2002,Galmiche2003,Galmiche05,Galmiche2019}, there has been comparatively little formal study on proof-search in the strict sense. One exception is the completeness result for (unit-simple) uniform proofs \cite{Armelin02} which is partially subsumed by the results herein.

The \emph{focusing principle} was introduced for Linear Logic \cite{Andreoli92} and is characterised by alternating \emph{focused} and \emph{unfocused} phases of goal-directed proof-search. The unfocused phase comprises rules which are safe to apply (i.e. rules where provability is invariant); conversely, the focused phase contains the reduction of a formula and its sub-formulas where potentially invalid sequents may arise, and backtracking may be required. 
During focused proof-search the unfocused phases are performed eagerly, followed by controlled goal-directed focused phases, until safe reductions are available again. We say that the focusing principle holds when every provable sequent has a focused proof. This alternation can be enforced by a mechanism based on a partition of the set of formulas into two classes, \emph{positive} and \emph{negative}, which correspond to safe behaviour on the left and right respectively; that is, for negative formulas provability is invariant with respect to the application of a right rule, and for positive formulas, of a left rule, but in the other cases the application may result in invalid sequents.

The original proof of the focusing principle in Linear Logic was via long and tedious permutations of rules \cite{Andreoli92}. In this paper, we use for BI a different methodology, originally presented in \cite{Laurent04}, which has since been implemented in a variety of logics \cite{Liang09,Chaudhuri16b,Chaudhuri16a} and proof systems \cite{Dyckhoff06}. The method is as follows: given a sequent calculus, first one polarises the syntax according to the positive/negative behaviours; second, one gives a focused variation of the sequent calculus where the control flow of proof-search is managed by polarisation; third, one shows that this system admits cut (the only non-analytic rule); and, finally, one shows that in the presence of cut the original sequent calculus may be simulated in the focused one. When the polarised system is complete, the focusing principle holds.  

In $\lbi$ certain rules (the structural rules) have no natural placement in either the focused or the unfocused phases of proof-search. Thus, a design choice must be made: to eliminate/constrain these rules, or to permit them without restriction.
The first gives a stricter control proof-search regime, but the latter typically achieves a more well-behaved proof theoretic meta-theory. In this paper, we choose the former as our motivation is to study computational behaviour of proof-search in BI, the latter being recovered by familiar admissibility results. The only case where confinement is not possible is the \emph{exchange} rule. In standard sequent calculi the exchange rule is made implicit by working with a more convenient data-structure such as multisets as opposed to lists; however, the specific structure of bunches in BI means that a more complex alternative is required. The solution presented is to use nested multisets of two types (additive and multiplicative) corresponding to the two different context-formers/conjunctions.

In Section~\ref{sec:BI} we present the logic of Bunched Implications; in particular, Section~\ref{sec:BI_trad} and Section~\ref{sec:BI_calc} contain the background on BI (the syntax and sequent calculus respectfully); meanwhile, Section~\ref{sec:BI_nests} gives representation of bunches as nested multisets. Section~\ref{sec:foc} contains the focused system: first, in Section~\ref{sec:foc_pol} we introduce the polarised syntax; second, in Section~\ref{sec:foc_calc} we introduce the focused sequents calculus and some metatheory, most importantly the $\cut$-admissibility result; finally, in Section~\ref{sec:foc_comp} we give the completeness theorem, from which the validity of the focusing principle follows as a corollary. We conclude in Section~\ref{sec:conclusion} with some further discussion and future directions. 

%% file: background.tex
	\subsection{Traditional Syntax} \label{sec:BI_trad}
	The logic BI has a well-studied metatheory admitting familiar categorical, algebraic, and truth-functional semantics which have the expected dualities~\cite{Pym2004resource,Galmiche05,Pym02,Docherty19,Pym19} . 
%	\todo{Problem with the references?}
	In practice, it is the free combination (or, more precisely, the fibration \cite{Gabbay1998,Pym02}) of intuitionistic logic (IL) and the multiplicative fragment of intuitionistic linear logic (MILL), which imposes the presence of two distinct context-formers in its sequent presentation. That is to say, the two conjunctions $\land$ and $*$ are represented at the meta-level by context-formers $;$ and $,$ in place of the usual commas for IL and MILL respectively. 
		\begin{definition}[Formula]
			Let \emph{\textsf{P}} be a denumerable set of propositional letters. The \emph{formulas} of BI, denoted by small Greek letters ($\phi, \psi, \chi, \ldots$), are defined by the following grammar, where $A \in \emph{\textsf{P}}$,
			$$\phi ::= \top \mid \bot \mid \i \mid A \mid (\phi \land \phi) \mid ( \phi \lor \phi )  \mid (\phi \to \phi) \mid (\phi * \phi) \mid (\phi \wand \phi)  $$
%		 Let $\circ$ be a connective. 
		 If $\circ \in \{ \land, \lor, \to, \top \}$ then it is an additive connective and if $\circ \in \{ *, \wand, \i \}$ then it is a  multiplicative connective. The set of all formulas is denoted $\Formulas$.
		\end{definition} 
		\begin{definition}[Bunch] 
			A \emph{bunch} is constructed from the following grammar, where $\phi \in \Formulas$,
%			\todo{we want $\Delta$ to be allowed to be empty, don't we?}
			$$\Delta ::= \phi \mid \emptyset_+ \mid \emptyset_\times \mid (\Delta ; \Delta) \mid (\Delta , \Delta)$$
			The symbols $\emptyset_+$ and $\emptyset_\times$ are the additive and multiplicative units respectively, and the symbols $;$ and $,$ are the additive and multiplicative context-formers respectively. A bunch is \emph{basic} if it is a formula, $\emptyset_+$, or $\emptyset_{\times}$ and \emph{complex} otherwise. The set of all bunches is denoted $\B$, the set of complex bunches with additive root context-former by $\B^+$, and the set of complex bunches with multiplicative root context-former by $\B^\times$.
		\end{definition}	
			For two bunches $\Delta, \Delta' \in \B$ if $\Delta'$ is a sub-tree of $\Delta$, it is called a \emph{sub-bunch}.  
			We may use the standard notation $\Delta(\Delta')$ (despite its slight inpracticality) to denote that $\Delta'$ is a sub-bunch of $\Delta$, in which case $\Delta(\Delta'')$ is the result of replacing the occurrence of $\Delta'$ by $\Delta''$. 
			If $\delta$ is a sub-bunch of $\Delta$, then the context-former $\circ$ is said to be its principal context-former in $\Delta(\Delta' \circ \delta)$ (and $\Delta(\delta \circ \Delta')$). 	
%			\sonia{why not the context-former at the root of $\delta$ itself?}
		\begin{example}\label{ex:abunch} Let $\phi$, $\psi$ and $\chi$ be formulas, and let $\Delta = (\phi,(\chi;\emptyset_+));(\psi;(\psi;\emptyset_\times))$. The bunch may be written for example as $\Delta(\phi,(\chi;\emptyset_+))$ which means that we can have $\Delta(\phi;\phi)=(\phi;\phi);(\psi;(\psi;\emptyset_\times))$. 
		\end{example}
	
	\begin{definition}[Bunched Sequent]
			A bunched sequent is a pair of a bunch $\Delta$, called the context, and a formula $\phi$, denoted $\Delta \seq \phi$.
	\end{definition}
	Bunches are intended to be considered up-to \emph{coherent equivalence} $(\equiv)$.
%	 analogously to other calculi where contexts are lists, but considered modulo permutation to become multisets. 
	 It is the least relation satisfying:
	\begin{itemize}
		\item Commutative monoid equations for $;$ with unit $\emptyset_+$,
		\item Commutative monoid equations for $,$ with unit $\emptyset_{\times}$,
		\item Congruence: if $\Delta'\equiv\Delta'' $ then $\Delta(\Delta') \equiv \Delta(\Delta'')$. 
	\end{itemize} 
	It will be useful to have a measure on sub-bunches which can identify their distance from the root node.
		\begin{definition}[Rank]
		If $\Delta'$ is a sub-bunch of $\Delta$, then $\rho(\Delta')$ is the number of alternations of additive and multiplicative context-formers between the principal context-former of $\Delta'$, and the root context-former of $\Delta$. 
%		\todo{alternation between what? additives and muliplicatives?}
	\end{definition}
	Let $\Delta$ be a complex bunch, we use $\Delta' \in \Delta$ to denote that $\Delta'$ is a (proper) top-most sub-bunch; that is, $\Delta$ is a sub-bunch satisfying $\Delta \neq \Delta'$ but $\rho(\Delta') = 0$.
\begin{example} Let $\Delta$ be as in Example \ref{ex:abunch}, then $\rho(\emptyset_+)=2$ whereas $\rho(\emptyset_\times)=0$; 
	hence, $\psi$, $\emptyset_{\times}$ and $(\phi,(\chi, \emptyset_{\times})) \in\Delta$. Consider the parse-tree of $\Delta$:
\[
\xymatrix@R=2mm@C=2mm{
		 &		&   & ; \ar@{-}[dll] \ar@{-}[dr]	  &	      &  &\\
		 & ,\ar@{-}[dl] \ar@{-}[dr]	&	& 			  & ; \ar@{-}[dl] \ar@{-}[dr]     &  &\\
	\phi & 		& ; \ar@{-}[dl] \ar@{-}[dr] & \psi		  &	  	  & ; \ar@{-}[dl] \ar@{-}[dr]&\\
		 &\chi 	&	& \emptyset_+ & \psi  &  & \emptyset_{\times}
}
\]
	
	Reading upward from $\emptyset_+$ one encounters first $;$ which changes into $,$ and then back to $;$ so the rank is $2$; whereas counting up from $\emptyset_\times$ one only encounters $;$ so the rank is $0$. 
\end{example}

		\subsection{Sequent Calculus} \label{sec:BI_calc} 
		 The proof theory of BI is well-developed including familiar  Hilbert, natural deduction, sequent calculi, tableaux systems, and display calculi ~\cite{Pym02,Galmiche05,Brotherston10a}. In the foregoing we restrict attention to the sequent calculus as it more amenable to studying proof-search as computation, having local correctness while enjoying the completeness of analytic proofs.
		
		\begin{definition}[System $\lbi$]
			The bunched sequent calculus $\lbi$ is composed of the rules in Figure \ref{fig:lbi}.
		\end{definition}
		The classification of $\land$ as additive may seem dubious upon reading the $\rrn\land$ rule, but the designation arises from the use of the structural rules; that is, the $\rrn\land$ and $\rrn\to$ rules may be replaced by \emph{additive} variants without loss of generality. The presentation in Figure \ref{fig:lbi} is as in \cite{Pym02} and simply highlights the nature of the additive and multiplicative context-formers. Nonetheless, the choice of rule does affect proof-search behaviours, and the consequences are discussed in more detailed in Section \ref{sec:foc_pol}.
		 
		 \LBI

	\begin{lemma}[Cut-elimination]\label{lem:LBI-cutelim}
		If $\phi$ has a $\lbi$-proof, then it has a $\cut$-free $\lbi$-proof, i.e., a proof with no occurence of the $\cut$ rule.
	\end{lemma}
	Throughout, unless specified otherwise, we take proof to mean $\cut$-free proof. Moreover, if $\mathsf{L}$ is a sequent calculus we use $\vdash_\mathsf{L }\Delta \seq \phi$ to denote that there is an $\mathsf{L}$-proof of $\Delta \seq \phi$. Further, if $\mathsf{R}$ is a rule, then we may denote $\mathsf{L+R}$ to denote the sequent calculus combining the rules of $\mathsf{L}$ with $\mathsf{R}$.
	
	 The following result, that a generalised version of the axiom is derivable in $\lbi$, will allow for such sequents to be used in proof-construction later on.
	
	\begin{lemma}\label{lem:formulaAx}
		For any formula $\phi$, $\vdash_{\lbi} \phi \seq \phi$.
	\end{lemma}
	\begin{proof}
		Follows from induction on size of $\phi$.
	\qed \end{proof}	

	The remainder of this section is the meta-theory required to control the structural rules, which pose the main issue to the study of proof-search in BI.

\begin{lemma}\label{lem:weakelim}
	The following rules are derivable in $\lbi$, and replacing $\weak$ with them does not affect the completeness of the system.
	$$
	\infer[\rn{Ax'}]{\Delta;A \seq A}{}
	\quad
	\infer[\rrn{{\i}'}]{\Delta;\emptyset_\times \seq \i}{}
	\quad
	\infer[\rrn{\top'}]{\Delta;\emptyset_+ \seq \top}{}
	$$
	%\begin{minipage}{0.32\linewidth}
	%	\begin{scprooftree}{0.9}
	%		\AxiomC{}
	%		\RightLabel{$Ax'$}
	%	\UnaryInfC{$\Delta;A \seq A$}
	%	\end{scprooftree}
	%\end{minipage}
	%\begin{minipage}{0.32\linewidth}
	%	\begin{scprooftree}{0.9}
	%	\AxiomC{}
	%	\RightLabel{$\i R'$}
	%	\UnaryInfC{$\Delta;\emptyset_\times \seq \i$}
	%\end{scprooftree}
	%\end{minipage}
	%\begin{minipage}{0.32\linewidth}
	%	\begin{scprooftree}{0.9}
	%	\AxiomC{}
	%	\RightLabel{$\top R '$}
	%	\UnaryInfC{$\Delta;\emptyset_+ \seq \top$}
	%\end{scprooftree}
	%\end{minipage}
	$$
	\infer[\rrn{\ast'}]{(\Delta,\Delta');\Delta'' \seq \phi*\psi }{\Delta \seq \phi & \Delta' \seq \psi}
	\quad
	\infer[\lrn{\wand'}]{\Delta(\Delta',\Delta'',(\Delta''';\phi \wand \psi)) \seq \chi }{\Delta' \seq \phi & \Delta(\Delta'', \psi) \seq \chi}
	$$
%	\begin{minipage}{0.45\linewidth}
%	\begin{scprooftree}{0.9}
%		\AxiomC{$\Delta \seq \phi$}
%		\AxiomC{$\Delta' \seq \psi$}
%		\RightLabel{$*R'$}
%		\BinaryInfC{$(\Delta,\Delta');\Delta'' \seq \phi*\psi $} 
%	\end{scprooftree}
%\end{minipage}
%	\begin{minipage}{0.45\linewidth}
%	\begin{scprooftree}{0.9}
%		\AxiomC{$\Delta' \seq \phi$}
%		\AxiomC{$\Delta(\Delta'', \psi) \seq \chi$}
%		\RightLabel{$\wand L'$}
%		\BinaryInfC{$\Delta(\Delta',\Delta'',(\Delta''';\phi \wand \psi)) \seq \chi $} 
%	\end{scprooftree}
%\end{minipage}\\
\end{lemma}
\begin{proof}
	We can construct in $\lbi$ derivations with the same premisses and conclusion as these rules by use of the structural rules. Let $\lbi'$ be $\lbi$ without $\weak$ but with these new rules (retaining also $\rrn\ast,\lrn\wand,\rrn\i,\rrn\top,$ and $\rn{Ax}$), then $\weak$ is admissible in $\lbi'$ using standard permutation argument.\qed \end{proof}

One may regard the above modification to $\lbi$ as forming a new calculus, but since all the new rules are derivable it is really a restriction of the calculus, in the sense that all proofs in the new system have equivalent proofs in $\lbi$ differing only by explicitly including instances of weakening.

\subsection{Nested Calculus} \label{sec:BI_nests}
Originally, sequents in the calculi for classical and intuitionistic logics (\textsf{LK} and \textsf{LJ}, respectively) were introduced as lists, and a formal \emph{exchange} rule was required to permute elements when needed for a logical rule to be applied~\cite{Gentzen1969}.  However, in practice, the exchange rule is often suppressed, and contexts are simply presented as multisets of formulas. This reduces the number of steps/choices being made during proof-search without increasing the complexity of the underlying data structure. Bunches have considerably more structure than lists, but a quotient with respect to coherent equivalence can be made resulting in two-sorted nested multisets; this was first suggested in \cite{Donnelly05}, though never formally realised.
	\begin{definition}[Two-sorted Nest]
		Nests $(\Gamma)$ are formulas or multisets, ascribed either additive $(\Sigma)$, or multiplicative $(\Pi)$ kind, containing nests of the opposite kind:
		\begin{align*}
		\Gamma := \Sigma \mid \Pi \qquad	\Sigma := \phi  \mid \{\Pi_1,...,\Pi_n\}_+  \qquad	\Pi := \phi \mid \{\Sigma_1,...,\Sigma_n\}_\times 
		\end{align*}
		The constructors are multiset constructors which may be empty in which case the nests are denoted $\varnothing_+$ and $\varnothing_\times$ respectively. No multiset is a singleton; and the set of all nests is denoted $\nests$.
	\end{definition}
	Given nests $\Lambda$ and $\Gamma$, we write $\Lambda \in \Gamma$ to denote either that $\Lambda =\Gamma$, if $\Gamma$ is a formula, or that $\Lambda$ is an element of the multiset $\Gamma$ otherwise. Furthermore, we write $\Lambda \subseteq \Gamma$ to denote $\forall \gamma \in \nests$ if $\gamma \in \Lambda$ then $\gamma \in \Gamma$. 
	
	We will depart from the standard, yet impractical subbunch notation, and adopt a context notation for nests instead.
	We write $\Gamma\{\cdot\}_+$ (resp. $\Gamma\{\cdot\}_\times$) for a nest with a hole within one of its additive (resp. multiplicative) multisets.% (but we may remove the subscript when possible).
%	\todo{reformulate the rest of this paragraph accordingly}
	The notation $\Gamma\{\Lambda\}_+$ (resp. $\Gamma\{\Lambda\}_\times$), denotes that $\Lambda$ is a sub-nest of $\Gamma$ of additive (resp. multiplicative) kind; we may use $\Gamma\{\Lambda\}$ when the kind is not specified. In either case $\Gamma\{\Lambda'\}$ denotes the substitution of $\Lambda$ for $\Lambda'$. A promotion in the syntax tree may be required after a substitution either to handle a singleton or an improper alternation of constructor types.
	
	\begin{example}\label{ex:nests} The following inclusions are valid,
		\[
		 \{\phi \,, \chi  \, \}_\times \in \Big \{ \, \{\phi \,, \chi  \, \}_\times , \psi  \Big \}_+ \subseteq \Big\{ \, \{\phi \,, \chi \, \}_\times , \psi \,, \psi \,, \emptyset_\times \, \Big\}_+ = \Gamma\{\{\phi \,, \chi  \, \}_\times\}_+
		\]
		It follow that $\Gamma\{\{\phi\,,\phi\}_+\}_+ = \{ \, \phi \,, \phi \,, \psi \,, \psi \,, \emptyset_\times \, \}_+$. Note the absence of the $\{ \cdot \}_+$ constructor after substitution, this is due to a promotion in the syntax tree to avoid having two nested additive constructors. Similarly, since $\emptyset_\times$ denotes the empty multiset of multiplicative kind, substituting $\chi$ with it gives $\{ \phi , \psi \,, \psi \,, \emptyset_\times \, \}_+ $; that is, first the improper $\{\phi, \emptyset_\times\}_\times$ becomes $\{\phi\}_\times$; then, the resulting singleton $\{\phi\}_\times$ is promoted to $\phi$.
	\end{example}
	Typically we will only be interested in fragments of sub-nests so we have the following abuse of notation, where $\circ \in \{+,\times\}$:
\[\Gamma\{\{\Pi_1,...,\Pi_i\}_\circ, \Pi_{i+1},..,\Pi_n\}_\circ := \Gamma\{\Pi_1,...,\Pi_n\}_\circ \]
	The notion of rank has a natural analogue in this setting.
	\begin{definition}[Depth, Rank]
		Let $\circ \in \{+\,,\times\}$ be a nest, we define the depth on $\B$ as follows:
		$$
		\delta(\phi) := 0 \qquad \delta(\{\Gamma_1,...,\Gamma_n\}_\circ) := \max\{\delta(\Gamma_1),...,\delta(\Gamma_n)\}+1 $$
	\end{definition}
	The equivalence of the two presentations, bunches and nests, follows from a moral (in the sense that bunches are intended to be considered modulo congruence) inverse between a \emph{nestifying} function $\eta$ and a \emph{bunching} function $\beta$. The transformation $\beta$ is simply going from a tree with arbitrary branching to a binary one, and $\eta$ is the reverse.
	
	\begin{definition}[Canonical Translation]
		The canonical translation $\nestify:\B \to \nests$ is defined recursively as follows,
		$$
		\nestify(\Delta) := 
		\begin{cases}
		\Delta & \text{if } \Delta \in \Formulas \cup \{\emptyset_+,\emptyset_\times\} \\
		\{\nestify(\Delta') \in \nests \mid \rho(\Delta')=1 \text{ and } \Delta' \in \B^\times \}_+ & \text{if } \Delta \in \B^+  \\
		\{\nestify(\Delta') \in \nests \mid \rho(\Delta')=1 \text{ and } \Delta' \in \B^+ \}_\times & \text{if } \Delta \in \B^\times 
		\end{cases}
		$$
%		\todo{does this actually work? $\Delta' \in \Delta$ does not necessarily mean that there is an alternation of context-former between $\Delta$ and $\Delta'$.}
		The canonical translation $\bunch:\nests \to \B$ is defined recursively as follows,
		$$
		\bunch(\Gamma) := 
		\begin{cases}
		\Gamma & \text{ if } \Gamma \in \Formulas \cup \{\emptyset_+, \emptyset_\times\} \\
		\bunch(\Pi_1);(\bunch(\Pi_2);...)  & \text{ if } \Gamma = \{\Pi_1, \Pi_2,...\}_+  \\
		\bunch(\Sigma_1),(\bunch(\Sigma_2),...)  & \text{ if } \Gamma = \{\Sigma_1, \Sigma_2,...\}_\times
		\end{cases}
		$$
	\end{definition} 
	\begin{example} Applying $\eta$ to the bunch in Example \ref{ex:abunch} gives the nest in Example~\ref{ex:nests}:
		$$
		\xymatrix@R=2mm@C=2mm{
		&	& & + \ar@{-}[dll] \ar@{-}[dl] \ar@{-}[dr] \ar@{-}[drr] & & \\
		&\times \ar@{-}[dl] \ar@{-}[dr]   &\psi& & \psi& \emptyset_\times \\
		\psi& &\chi		& 	& 	&
		}
		$$
	\end{example}
	\begin{lemma} \label{lem:transinverse}
		The translations are inverses up-to congruence; that is,
		\begin{enumerate}
			\item if $\Delta \in \B$ then $(\bunch \circ \eta)(\Delta) \equiv \Delta$;
			\item if $\Gamma \in \nests$ then $(\nestify \circ \bunch)(\Gamma) \equiv \Gamma$;
			\item let $\Delta, \Delta' \in \B$, then $\Delta \equiv \Delta'$ if and only if $\nestify(\Delta) = \nestify(\Delta')$.
		\end{enumerate} 
	\end{lemma}
\begin{proof}
	The first two statements follow by induction on the depth (either for bunches or nests), where one must take care to consider the case of a context consisting entirely of units. The third statement employs the first in the forward direction, and proceeds by induction on depth in the reverse direction. 
\qed \end{proof}
\begin{definition}[System $\eta\lbi$]
	The nested sequent calculus $\eta\lbi$ is composed of the rules in Figure~\ref{fig:taulbi}, where the metavariables denote possibly empty nests.
\end{definition}
Observe the use of metavariable $\Gamma'$ instead of $\Pi$ (resp. $\Sigma$) as sub-contexts in Figure~\ref{fig:taulbi}. This allows classes of inferences such as
\[
\infer[\rrn\ast]{\{\Sigma_0,...,\Sigma_n \}_\times \seq \phi *\psi}{
	\{\Sigma_0,...,\Sigma_i\}_\times \seq \phi
	&
	\{\Sigma_{i+1},...,\Sigma_n \}_\times \seq \phi
	}
\]
%\begin{prooftree}
%	\AxiomC{$\{\Sigma_0,...,\Sigma_i\}_\times \seq \phi$}
%	\AxiomC{$\{\Sigma_{i+1},...,\Sigma_n \}_\times \seq \phi$}
%	\BinaryInfC{$\{\Sigma_0,...,\Sigma_n \}_\times \seq \phi *\psi$}
%\end{prooftree}
to be captured by a single figure.  In practice it implements the abuse of notation given above: $$\{\{\Sigma_0,...,\Sigma_i\}_\times, \{\Sigma_{i+1},...,\Sigma_n\}_\times \}_\times  \seq \phi *\psi$$ 

\etalbi

This system	is a new and very convenient presentation of $\lbi$, not \emph{per se} a development of the proof theory for the logic. 
\begin{lemma}[Soundness and Completeness of $\nestify \lbi$] \label{lem:etalbi}
	Systems $\lbi$ and $\eta\lbi$ are equivalent:
	\begin{enumerate}
		\item[] Soundness: If $\vdash_{\eta\lbi} \Gamma \seq \phi$  then $\vdash_{\lbi} \beta(\Gamma) \seq \phi$;
		\item[] Completeness: If $\vdash_{\lbi} \Delta \seq \phi$  then $\vdash_{\eta\lbi} \nestify(\Delta) \seq \phi$.
	\end{enumerate}
\end{lemma}
\begin{proof}
Each claim follows by induction on the context, appealing to Lemma \ref{lem:transinverse} to organise the data structure for the induction hypothesis, without loss of generality. 
\end{proof}
\begin{example}\label{ex:nestedproof}
	The following is a proof in $\eta \lbi$. 
%	using $W$ and $Ax$ to close the leaves; 
%	more precisely, it is a proof in the complete sub-system arising in Lemma \ref{lem:weakelim} using nested sequents:
%	\begin{scprooftree}{0.9}
%		\AxiomC{$A \seq A$}
%		\AxiomC{$\{B, C\}_+ \seq B $}
%		\BinaryInfC{$\{A , \{B , C\}_+ \}_\times \seq A * B$ }
%		\UnaryInfC{$\{A , (B \land C) \}_\times \seq A * B$ }
%		\AxiomC{$A \seq A$}
%		\AxiomC{$\{B, C\}_+ \seq C$}
%		\UnaryInfC{$B\land C \seq C$ }
%		\BinaryInfC{$\{A , (B \land C) \}_\times \seq A * C$ }
%		\BinaryInfC{$\{A , (B \land C) \}_\times \seq (A * B)\land (A * C)$ }
%		\UnaryInfC{$A * (B \land C)  \seq (A * B)\land (A * C)$}
%		\UnaryInfC{$\emptyset_\times \seq (A * (B \land C)) \wand ((A * B)\land (A * C))$}
%	\end{scprooftree}
	\[\scalebox{.9}{$
	\infer[\rrn\wand]{\emptyset_\times \seq (A * (B \land C)) \wand ((A * B)\land (A * C))}{
		\infer[\lrn\ast]{A * (B \land C)  \seq (A * B)\land (A * C)}{	
			\infer[\rrn\land]{\{A , (B \land C) \}_\times \seq (A * B)\land (A * C)}{
				\infer[\lrn\land]{\{A , (B \land C) \}_\times \seq A * B}{
					\infer[\rrn\ast]{\{A , \{B , C\}_+ \}_\times \seq A * B}{
						\infer[\rn{Ax}]{A \seq A}{}
						&
						\infer[\rn{Ax}]{\{B, C\}_+ \seq B }{}
						}
					}
				&
				\infer[\rrn\ast]{\{A , (B \land C) \}_\times \seq A * C}{
					\infer[\rn{Ax}]{A \seq A}{}
					&
					\infer[\lrn\land]{B\land C \seq C}{
						\infer[\rn{Ax}]{\{B, C\}_+ \seq C}{}
						}
					}
				}
			}
		}
	$}\]
\end{example}
We expect no obvious difficulty in studying focused proof-search with bunches instead of nested multisets; the design choice is simply to reduce the complexity of the argument by pushing all uses of exchange (\rn{E}) to Lemma~\ref{lem:etalbi}, rather than tackle it at the same time as focusing itself. In particular, working without the nested system would mean working with a weaker notion of focusing since the exchange rule must then be permissible during both focused and unfocused phases of reduction.

%% file: focusedsystem.tex
At no point in this section will we refer to bunches, thus the variable $\Delta$, so far reserved for elements of $\B$, is re-appropriated as an alternative to $\Gamma$.

 	\subsection{Polarisation} \label{sec:foc_pol}
 	
 	 Polarity in the focusing principle is determined by the invariance of provability under application of a rule, that is, by the proof rules themselves. One way the distinction between positive and negative connectives is apparent is when their rule behave either \emph{synchronously} or \emph{asynchronously}. For example, the $\rrn\ast$ and $\lrn\wand$ highlight the {synchronous} behaviour of the multiplicative connectives since the structure of the context affects the applicability of the rule. Displaying such a synchronous behaviour on the left makes $\wand$ a negative connective, while having it on the right makes $\ast$ a positive connective. 
 	 
 Another way to characterise the polarity of a connective is the study of the inveribility properties of the corresponding rules.
% 
% Since each connective has one right rule and one left rule, polarity can simply be assigned by case analysis.
%\begin{example}\label{ex:invertibilityofandright}
	For example, consider the inverses of the $\lrn\lor$ rule, 
	\[\scalebox{.9}{$
	\infer[\text{$\lrn[1]{\lor^{inv}}$}]{\Gamma\{\phi\} \seq \chi}{\Gamma\{\phi\lor\psi\} \seq \chi }
	\qquad\qquad
	\infer[\text{$\lrn[2]{\lor^{inv}}$}]{\Gamma\{\psi\} \seq \chi}{\Gamma\{\phi\lor\psi\} \seq \chi}
	$}\]
%	\begin{minipage}{0.45\linewidth}
%		\centering
%		\begin{scprooftree}{0.8}
%			\AxiomC{$\Gamma\{\phi\lor\psi\} \seq \chi $}
%			\UnaryInfC{$\Gamma\{\phi\} \seq \chi$}
%		\end{scprooftree}
%	\end{minipage}
%	\begin{minipage}{0.45\linewidth}
%		\centering
%		\begin{scprooftree}{0.8}
%			\AxiomC{$\Gamma\{\phi\lor\psi\} \seq \chi $}
%			\UnaryInfC{$\Gamma\{\psi\} \seq \chi$}
%		\end{scprooftree}
%	\end{minipage}\\
%	\vspace{0.2cm}
	They are derivable in $\lbi$ with $\cut$ (below -- the left branch being closed using Lemma~\ref{lem:formulaAx}) and therefore admissible in $\lbi$ without $\cut$ (by Lemma~\ref{lem:LBI-cutelim}). 
	\[\scalebox{.9}{$
	\infer[\cut]{\Gamma\{\phi\} \seq \chi}{
		\infer[\rrn\lor]{\phi \seq \phi\lor\psi}{
			\phi \seq \phi
			}
		&
		\Gamma\{\phi\lor\psi\} \seq \chi
		}
	\qquad
	\infer[\cut]{\Gamma\{\psi\} \seq \chi}{
		\infer[\rrn\lor]{\psi \seq \phi\lor\psi}{
			\psi \seq \psi
		}
		&
		\Gamma\{\phi\lor\psi\} \seq \chi
	}
	$}\]
%	\begin{minipage}{0.45\linewidth}
%		\centering
%		\begin{scprooftree}{0.8}
%			%			\AxiomC{$\vdots$}
%			%			\noLine
%			\AxiomC{$\phi \seq \phi$}
%			\UnaryInfC{$\phi \seq \phi\lor\psi$}
%			\AxiomC{$\Gamma\{\phi\lor\psi\} \seq \chi$}
%			\RightLabel{$\cut$}
%			\BinaryInfC{$\Gamma\{\phi\} \seq \chi$}
%		\end{scprooftree}
%	\end{minipage}
%	\begin{minipage}{0.45\linewidth}
%		\centering
%		\begin{scprooftree}{0.8}
%%			\AxiomC{$\vdots$}
%%			\noLine
%			\AxiomC{$\psi \seq \psi$}
%			\UnaryInfC{$\psi \seq \phi\lor\psi$}
%			\AxiomC{$\Gamma\{\psi\lor\psi\} \seq \chi$}
%			\RightLabel{$\cut$}
%			\BinaryInfC{$\Gamma\{\psi\} \seq \chi$}
%		\end{scprooftree}
%	\end{minipage}\\
	This means that provability is invariant in general upon application of $\lrn\lor$ since it can always be reverted if needed, as follows
	\[\scalebox{.9}{$
	\infer[\lrn\lor]{\Gamma\{\phi\lor\psi\} \seq \chi}{
		\infer[\text{$\lrn[1]{\lor^{inv}}$}]{\Gamma\{\phi\} \seq \chi}{
			\Gamma\{\phi\lor\psi\} \seq \chi
			}
		&
		\infer[\text{$\lrn[2]{\lor^{inv}}$}]{\Gamma\{\psi\} \seq \chi}{
			\Gamma\{\phi\lor\psi\} \seq \chi
			}
		}
	$}
	\]
%	 	\begin{minipage}{\linewidth}
%		\centering
%		\begin{scprooftree}{0.8}
%			\AxiomC{$\Gamma\{\phi\lor\psi\} \seq \chi$}
%			\UnaryInfC{$\Gamma\{\phi\} \seq \chi$}
%			\AxiomC{$\Gamma\{\phi\lor\psi\} \seq \chi$}
%			\UnaryInfC{$\Gamma\{\psi\} \seq \chi$}
%			\RightLabel{$\lor L$}
%			\BinaryInfC{$\Gamma\{\phi\lor\psi\} \seq \chi$}
%		\end{scprooftree}
%	\end{minipage}
%	\vspace{0.2cm}

%This is one way of justifying the assignment of polarity to connectives. 
Note however that dual connectives do not necessarily have dual behaviours in terms of provability invariance, on the left and on the right. For example, consider all the possible rules for $\land$, of which some qualify as positive and others as positive.
% 
%\begin{minipage}{\linewidth}
%	\begin{minipage}{0.45\linewidth}
%		\begin{scprooftree}{0.8}
%			\AxiomC{$\Gamma\{\phi_i\} \seq \chi$}
%			\RightLabel{$\land^- L$}
%			\UnaryInfC{$\Gamma\{\phi_1\land \phi_0\} \seq \chi$}
%		\end{scprooftree}
%	\end{minipage}
%	\begin{minipage}{0.45\linewidth}
%	\begin{scprooftree}{0.8}
%		\AxiomC{$\Gamma \seq \phi$}
%		\AxiomC{$\Gamma \seq \psi$}
%					\RightLabel{$\land^- R$}
%		\BinaryInfC{$\Gamma \seq \phi \land \psi$}
%	\end{scprooftree}
%	\end{minipage}
	\[\scalebox{.9}{$
	\begin{array}{c@{\quad}c}
		\infer[\text{$\lrn[1]{\land^-}$}]{\Gamma\{\phi\land \psi\} \seq \chi}{\Gamma\{\phi\} \seq \chi}
		\quad 
		\infer[\text{$\lrn[2]{\land^-}$}]{\Gamma\{\phi\land \psi\} \seq \chi}{\Gamma\{\psi\} \seq \chi}
		&
		\infer[\rrn{\land^- }]{\Gamma \seq \phi \land \psi}{\Gamma \seq \phi &\Gamma \seq \psi }
		\\[2ex]
		\infer[\lrn{\land^+}]{\Gamma\{\phi\land\psi\} \seq\chi}{\Gamma\{\{\phi,\psi\}_+\} \seq\chi}
		&
		\infer[\rrn{\land^+}]{\{\Gamma, \Delta \}_+ \seq \phi \land \psi}{\Gamma \seq \phi & \Gamma\{\{\phi\land\psi\}_+\} \seq\chi}
	\end{array}
	$}\]
%
%	\begin{minipage}{0.45\linewidth}
%		\begin{scprooftree}{0.8}
%			\AxiomC{$\Gamma\{\{\phi,\psi\}_+\} \seq\chi$}
%						\RightLabel{$\land^+ L$}
%			\UnaryInfC{$\Gamma\{\{\phi\land\psi\}_+\} \seq\chi$}
%		\end{scprooftree}
%	\end{minipage}
%	\begin{minipage}{0.45\linewidth}
%		\begin{scprooftree}{0.8}
%			\AxiomC{$\Gamma \seq \phi$}
%			\AxiomC{$\Delta \seq \psi$}
%						\RightLabel{$\land^+ R$}
%			\BinaryInfC{$\{\Gamma, \Delta \}_+ \seq \phi \land \psi$}
%		\end{scprooftree}
%	\end{minipage}
%%\end{minipage}\\
%\vspace{0.2cm}

 All of these rules are sound, and replacing the conjunction rules in $\lbi$ with any pair of a left and right rule will result in a sound and complete system. Indeed, the rules are inter-derivable when the structural rules are present, but otherwise they can be paired to form two sets of rules which have essentially different proof-search behaviours. That is, the rules in the top-row make $\land$ negative while the bottom row make $\land$ positive. Each conjunction also comes with an associated unit, that is, $\top^-$ for negative conjunctio and $\top^+$ for positive conjunction.
 We choose to add  all of them to our system in order to have access to those different proof search behaviours at will.
 
  Finally, the polarity of the propositional letters can be assigned arbitrarily as long as only once for each.
 
\begin{definition}[Polarised Syntax]
 	Let $\mathsf{P}^+ \sqcup \mathsf{P}^-$ be a partition of $\mathsf{P}$, and let $A^+ \in \mathsf{P}^+$ and $A^- \in \mathsf{P}^-$, then the polarised formulas are defined by the following grammar,
	\begin{align*}
	P,Q &::=  L \mid P \lor Q \, \, \mid \, P * Q \, \,  \, \, \mid P \land^+ Q \, \mid \top^+ \mid \i \mid \bot  &L ::= \downshift N \mid A^+ \\
	N,M &::= R \mid P \to N \mid P \wand N \mid N \land^- M  \mid \top^-  &R ::=  \upshift P \mid A^-
	\end{align*}
	The set of positive formulas $P$ is denoted $\Formulas^+$; the set of negative formulas $N$ is denoted $\Formulas^-$; and the set of all polarised formulas is denoted $\Formulas^\pm$. The sub-classifications $L$ and $R$ are left-neutral and right-neutral formulas respectfully. 
\end{definition}
The shift operators have no logical meaning; they simply mediate the exchange of polarity, and thus the \emph{shifting} into a new phase of proof-search. Consequently, to reduces cases in subsequent proofs, we will consider formulas of the form $\upshift \downshift N$ and $\downshift \upshift P$, but not $ \downshift \upshift \downshift N$, $\downshift \upshift \downshift \upshift P$, etc. 
%\todo{answer reviewer's comment here. maybe in a footnote}
\begin{definition}[Depolarisation]
	Let $\circ \in \{\lor\,, *\,, \to\,, \wand\}$, and let $A^+ \in \mathsf{P}^+$ and $A^- \in \mathsf{P}^-$, then the depolarisation function $\floor{\cdot}:\Formulas^\pm \to \Formulas$ is defined as follows: 
	\[\begin{array}{l}
	\floor{A^+} := \floor{A^-} := A 
	\quad 
	\floor{\upshift \phi} := \floor{\downshift \phi} := \floor{\phi} 
	\quad 
	\floor{\bot} := \bot 
	\quad\floor{\i} := \i 
	\\  
	\floor{\top^+} := \floor{\top^-} := \top 
	\quad 
	\floor{\phi \circ \psi} := \floor{\phi} \circ \floor{\psi}  
	\quad
	\floor{\phi \land^+ \psi} := \floor{\phi \land^- \psi} := \floor{\phi} \land \floor{\psi} 
	\end{array}\]
\end{definition}
Since proof-search is controlled by polarity, the construction of sequents in the focused system must be handled carefully to avoid ambiguity.
\begin{definition}[Polarised Sequents]
%	\emph{Positive/neutral/focused} nests, denoted by $\Gamma^+$/$\vec{\Gamma}$/$\hat{\Gamma}$, are defined according to the following grammars
\emph{Positive} and \emph{neutral} nests, denoted by $\Gamma$ and $\vec{\Gamma}$ resp., are defined according to the following grammars
\[\begin{array}{l@{\ :=\ }l@{\ \mid\ }l@{\qquad}l@{\ :=\ }l@{\qquad}l@{\ :=\ }l@{\quad}}
	 \Gamma& \Sigma & \Pi
	 &
	 \Sigma & P  \mid \{\Pi_1,...,\Pi_n\}_+  
	 &
	 \Pi& P  \mid \{\Sigma_1,...,\Sigma_n\}_\times  
	 \\
	 \vec{\Gamma}&\vec{\Sigma} & \vec{\Pi} 
	 &
	 \vec{\Sigma} & L \mid \{\vec{\Pi}_1,...,\vec{\Pi}_n\}_+  
	 &
	 \vec{\Pi}& L \mid \{\vec{\Sigma}_1,...,\vec{\Sigma}_n\}_\times  
%	 \\
%	 \hat{\Gamma}& \hat{\Sigma} & \hat{\Pi} 
%	 &
%	 \hat{\Sigma} & \foc{N} \mid  \{\hat{\Pi}, \vec{\Pi}_1,...,\vec{\Pi}_n\}_+  
%	 &
%	 \hat{\Pi}& \foc{N}  \mid \{\hat{\Sigma}, \vec{\Sigma}_1\,...,\vec{\Sigma}_n\}_\times
 \end{array}\]
	A pair of a polarised nest and a polarised formula is a \emph{polarised sequent} if it falls into one of the following cases
	$$
	\Gamma \seq N \quad \mid \quad \vec{\Gamma} \seq \foc{P} \quad \mid  \quad \vec{\Gamma}\{\foc{N}\} \seq R
	$$
\end{definition}
The decoration $\foc{\phi}$ indicates that the formula is in focus; that is, it is a positive formula on the right, or a negative formula on the left. Of the three possible cases for well-formed polarised sequents, the first may be called \emph{unfocused}, with the particular case of being \emph{neutral} when of the form $\vec{\Gamma} \seq R$; and the latter two may be called \emph{focused}. 

	 \begin{definition}[Depolarised Nest]
	 	The depolarisation map extends to polarised nests $\floor{\cdot}:\nests^\pm \to \nests$ as follows:
	 	$$\floor{\{\Pi_1,...,\Pi_n\}_+} = \{\floor{\Pi_1},...,\floor{\Pi_n} \}_+ \qquad \floor{\{\Sigma_1,...,\Sigma_n\}_\times} = \{\floor{\Sigma_1},...,\floor{\Sigma_n} \}_\times
	 	$$
	 \end{definition}
	 
\subsection{Focused Calculus} \label{sec:foc_calc}
We may now give the focused system. That is, the operational semantics for focused proof-search in $\lbi$. All the rules, with the exception of $\rn P$ and $\rn N$, are polarised versions of the rules from $\nestify\lbi$.

\focbi

\begin{definition}[System $\fbi$]
	The focused system $\fbi$ is composed of the rules on Figure \ref{fig:fbi}.
\end{definition}

Note the absence of a $\cut$-rule, this is because the above system is intended to encapsulate precisely \emph{focused} proof-search. Below we show that a $\cut$-rule is indeed admissible, but proofs in $\fbi+\cut$ are not necessarily focused themselves. Here the distinction between the methodologies for establishing the focusing principle becomes present since one may show completeness without leaving $\fbi$ by a permutation argument instead of a $\cut$-elimination one. 

The $\rn P$ and $\rn N$ rules will allow us to move a formula from one side to another during the proof of the completeness of $\fbi+\cut$ (Lemma~\ref{lem:compFBIcut}).%: the simulation of arbitrary proofs in the focused system when $\cut$ is present. 
%\todo{check if this is enough to answer the reviewer }
The depolarised version are not directly present in $\lbi$, but are derivable in $\lbi$ (Lemma \ref{lem:formulaAx}). However, the way they are focused renders them not provable in $\fbi$ because it forces one to begin with a potentially \emph{bad} choice; for example,  $A \lor B \seq A \lor B$ has no proof beginning with $\rrn\lor$. In practice, they are a feature rather than a bug since they allow one to terminate proof-search early, without unnecessary further expansion of the axiom.
In related works, such as~\cite{Chaudhuri16a,Chaudhuri16b}, the analogous rules are eliminated by initially working with a weaker notion of focused proof-search, and it is reasonable to suppose that the same may be true for BI. We leave this to future investigation.

Note also that, although it is perhaps proof-theoretically displeasing to incorporate weakening into the operational rules as in $\lrn{\wand'}$ and $\rrn{\ast'}$, it has good computational behaviour during focused proof-search since the reduction of $\phi \wand \psi$ can only arise out of an explicit choice made earlier in the computation.

Soundness follows immediately from the depolarisation map; that is, the interpretation of polarised sequents as nested sequents, and hence proofs in $\fbi$ actually are focused proofs in $\nestify\lbi$.

\begin{theorem}[Soundness of $\fbi$]
	Let $\Gamma$ be a polarised nest and $N$ a negative formula.
	If $\vdash_{\fbi} \Gamma \seq N$  then $\vdash_{\nestify\lbi} \floor{\Gamma} \seq \floor{N}$  
\end{theorem}

\begin{proof}
	Every rule in $\fbi$ except the shift rules, as well as the $\rn P$ and $\rn N$ axioms, become a rule in $\nestify \bi$ when the antecedent(s) and consequent are depolarised. Instance of the shift rule can be ignored since the depolarised versions of the consequent and antecedents are the same. Finally, the depolarised versions of $\rn P$ and $\rn N$ follow from Lemma \ref{lem:formulaAx} with the use of some weakening.\qed \end{proof}

\begin{example}
	Consider the following proof in $\fbi$, we suppose here that propositional letters $A$ and $C$ are negative, but $B$ is positive.
\[\scalebox{.9}{$
%\infer[label]{\emptyset_\times \seq \upshift \downshift ((\downshift A * \downshift(\upshift B\land^- C)) \wand (\upshift (\downshift A * B)\land^- \upshift (\downshift A * \downshift C)))}{
%	\infer[label]{\emptyset_\times \seq \foc{\downshift  ((\downshift A * \downshift(\upshift B\land^- C)) \wand (\upshift (\downshift A * B)\land^- \upshift (\downshift A * \downshift C))}}{
		\infer[\rrn\wand]{\emptyset_\times \seq  (\downshift A * \downshift(\upshift B\land^- C)) \wand (\upshift (\downshift A * B)\land^- \upshift (\downshift A * \downshift C))}{
			\infer[\lrn\ast]{\downshift A * \downshift(\upshift B\land^- C) \seq \upshift (\downshift A * B)\land^- \upshift (\downshift A * \downshift C)}{
				\infer[\rrn{\land^-}]{\{ \downshift A , \downshift(\upshift B\land^- C) \}_\times \seq \upshift (\downshift A * B)\land^- \upshift (\downshift A * \downshift C)}{
					\infer[\lrn\upshift\ {\color{blue}(1)}]{\{ \downshift A , \downshift(\upshift B\land^- C) \}_\times \seq \upshift (\downshift A * B)}{
						\infer[\text{$\lrn[1]{\land^-}$}]{\{ \downshift A , \foc{\upshift B\land^- C} \}_\times \seq \upshift (\downshift A * B)}{
							\infer[\lrn\upshift]{\{ \downshift A , \foc{\upshift B}\}_\times  \seq \upshift (\downshift A * B)}{
								\infer[\rrn\upshift]{\{ \downshift A , \ B\}_\times  \seq \upshift (\downshift A * B) }{
									\infer[\rrn\ast]{\{ \downshift A ,  B\}_\times  \seq \foc{\downshift A * B} }{
										\infer[\rrn\downshift]{\downshift A  \seq \foc{\downshift A }}{
											\infer[\lrn\downshift]{\downshift A  \seq A }{
												\infer[\rn{Ax^-}]{\foc{A}  \seq A }{}
											}
										}
										&
										\infer[\rn{Ax^+}]{B \seq \foc{B}}{}
										}
									}
								}
							}
						}
					&
					\infer[\rrn\upshift\ {\color{blue}(2)}]{\{ \downshift A , \downshift(\upshift B\land^- C) \}_\times \seq \upshift (\downshift A * \downshift C)}{
						\infer[\rrn\ast]{\{ \downshift A , \downshift(\upshift B\land^- C) \}_\times \seq \foc{\downshift A * \downshift C}}{
							\infer[\rrn\downshift]{\downshift A  \seq \foc{\downshift A }}{
								\infer[\lrn\downshift]{\downshift A  \seq A }{
									\infer[\rn{Ax^-}]{\foc{A}  \seq A }{}
									}
								}
							&
							\infer[\rrn\downshift]{\downshift(\upshift B\land^- C) \}_\times \seq \foc{\downshift C }}{
								\infer[\lrn\downshift]{\downshift(\upshift B\land^- C) \seq C}{
									\infer[\text{$\lrn[2]{\land^-}$}]{ \foc{\upshift B\land^- C}  \seq C}{
										\infer[\rn{Ax^-}]{ \foc{C} \seq C }{}
										}
									}
								}
							}
						}
					}
				}
			}
%		}
%	}
$}\]

	It is a focused version of the proof given in Example~\ref{ex:nestedproof}. Observe that the only non-deterministic choices are which formula to focus on, such as in steps (1) and (2), where different choices have been made for the sake of demonstration. The point of focusing is that \emph{only} at such points do choices that affect termination occur.  The assignment of polarity to the propositional letters is what forced the shape of the proof; for example, if $B$ had been negative the above would not have been well-formed. This phenomenon is standarly observed in focused systems (e.g.~\cite{Chaudhuri06}).
\end{example}

We now introduce the tool which will allow us to show that if there is a proof of a sequent (\emph{a priori} unstructured), then there is necessarily a focused one.
\begin{definition}
	All instances of the following rule where the sequents are well-formed are instances of $\cut$, where $\vec{\phi}$ denotes that $\phi$ is possibly prenexed with an additional shift
	\[
	\infer[\cut]{\Gamma\{\Delta\} \seq \chi}{\Delta \seq \phi & \Gamma \{\vec{\phi}\} \seq \chi}
	\]
%	\begin{minipage}{\linewidth}
%		\begin{scprooftree}{0.8}
%			\AxiomC{$\Delta \seq \phi$}
%			\AxiomC{$\Gamma \{\vec{\phi}\} \seq \chi$}
%			\RightLabel{$\cut$}
%			\BinaryInfC{$\Gamma\{\Delta\} \seq \chi$}
%		\end{scprooftree}
%	\end{minipage}
\end{definition}
Admissibility follows from the usual argument, but within the focused system; that is, through the upward permutation of cuts until they are eliminated in the axioms or are reduced in some other measure.
\begin{definition}[Good and Bad Cuts]
	Let $\mathcal{D}$ be a $\fbi+\cut$ proof, a cut is a quadruple $\langle \mathcal{L}, \mathcal{R}, \mathcal{C}, \phi \rangle$ where $\mathcal{L}$ and $\mathcal{R}$ are the premises to a $\cut$ rule, concluding $\mathcal{C}$ in $\D$, and $\phi$ is the $\cut$-formula. They are classified as follows:
	\begin{enumerate}
		\item[] \textbf{Good} - If $\phi$ is principal in both $\mathcal{L}$ and $\mathcal{R}$.
		\item[] \textbf{Bad} - If $\phi$ is not principal in one of $\mathcal{L}$ and $\mathcal{R}$.
		\begin{enumerate}
			\item[] Type 1: If $\phi$ is not principal in $\mathcal{L}$.
			\item[] Type 2: If $\phi$ is not principal in $\mathcal{R}$.
		\end{enumerate}
	\end{enumerate}
\end{definition}
\begin{definition}[Cut Ordering]
	The $\cut$-rank of a cut $\langle \mathcal{L}, \mathcal{R}, \mathcal{C}, \phi \rangle$ in a proof is the triple $\langle \cut$-complexity, $\cut$-duplicity, $\cut$-level$\rangle$, where the $\cut$-complexity is the size of $\phi$, the $\cut$-duplicity is the number of contraction instances above the cut, the $\cut$-level is the sum of the heights of the sub-proofs concluding $\mathcal{L}$ and $\mathcal{R}$. 
	
	Let $\D$ and $\D'$ be two $\fbi+\cut$ proofs, let $\sigma$ and $\sigma'$ denote their multiset of cuts respectively. Proofs are ordered by $\D \prec\D' \iff \sigma < \sigma'$, where $<$ is the multiset ordering derived from the lexicographic ordering on $\cut$-rank.
\end{definition} 	
It follows from a result in \cite{Dershowitz79} that the ordering on proofs is a well-order, since the ordering on cuts is a well-order. 
\begin{lemma}[Good Cuts Elimination]\label{lem:goodcutelim}
	Let $\D$ be a $\fbi+\cut$ proof of $S$; there is a $\fbi+\cut$ proof $\D'$ of $S$ containing no good cuts such that $\D' \preceq \D$.
\end{lemma}
\begin{proof}
	Let $\D$ be as in hypothesis, if it contains no good cuts then $\D =\D'$ gives the desired proof. Otherwise, there is at least one good cut $\langle \mathcal{L}, \mathcal{R}, \mathcal{C}, \phi \rangle$. Let $\partial$ be the sub-proof in $\D$ concluding $\mathcal{C}$, then there is a transformation $\partial \mapsto \partial'$ where $\partial'$ is a $\fbi+\cut$ proof of $S$ with $\partial' \prec \partial$ such that the multiset of good cuts in $\partial'$ is smaller (with respect to $\prec$) than the multiset of good cuts in $\partial$.  Since $\prec$ is a well-order indefinitely replacing $\partial$ with $\partial'$ in $\D$ for various cuts yields the desired $\D'$. 
	
	The key step is that a cut of a certain $\cut$-complexity is replaced by cuts of lower $\cut$-complexity, possibly increasing the $\cut$-duplicity or $\cut$-level of other cuts in the proof, but not modifying their complexity.
	
	\scalebox{.85}{$
	\infer[\cut]{\vec{\Gamma}\{\{\vec{\Gamma}' , A^+\}_+\} \seq \foc{A^+}}{
		\infer[\rn{Ax^+}]{ \{\vec{\Gamma}', A^+\}_+ \seq \foc{A^+} }{}
		&
		 \vec{\Gamma}\{A^+\} \seq \foc{A^+} 
		}
	$}
	$\quad\mapsto\quad$
	\scalebox{.85}{$
	\infer=[\weak]{\vec{\Gamma}\{\{\vec{\Gamma}' , A^+\}_+\} \seq \foc{A^+}}{\vec{\Gamma}\{A^+\} \seq \foc{A^+} }
	$}\\
	
%		\begin{minipage}{\linewidth}
%		\begin{minipage}{0.5\linewidth}
%			\begin{scprooftree}{0.8}
%				\AxiomC{$ \{\vec{\Delta}, \downshift \upshift A^+\}_+ \seq \foc{A^+} $}
%				\AxiomC{$\{\vec{\Delta}',A^+ \}_+ \seq \foc{A^+}$}
%				\BinaryInfC{$\{\vec{\Delta}, \vec{\Delta}', \downshift \upshift A^+ \}_+ \seq \foc{A^+}$}
%			\end{scprooftree}
%		\end{minipage} 
%		\begin{minipage}{0.1\linewidth}
%			\centering
%			\scalebox{1.5}{$\mapsto$}
%		\end{minipage}
%		\begin{minipage}{0.35\linewidth}
%			\begin{scprooftree}{0.8}
%				\AxiomC{}
%				\UnaryInfC{$\{\vec{\Delta}, \vec{\Delta}', \downshift \upshift A_+ \}_+ \seq \foc{A_+}$}
%			\end{scprooftree}
%		\end{minipage}
%	\end{minipage}
	
	\scalebox{.85}{$
	\infer[\cut]{\vec\Gamma\{\vec{\Delta},\vec{\Delta}',\{\vec{\Delta}'',\vec{\Delta}'''\}_+\}_\times \seq R}{
		\infer[\rrn\wand]{\vec{\Delta}''' \seq P\wand N}{
			\{\vec{\Delta}''', P\}_\times \seq  N
		}
		&
		\infer[\lrn\wand]{\vec\Gamma\{\vec{\Delta},\vec{\Delta}',\{\vec{\Delta}'',\foc{P\wand N}\}_+\}_\times \seq R}{
			\vec{\Delta} \seq \foc{P}
			&
			\Gamma\{\vec{\Delta}', \foc{N}\}_\times \seq R
			}
		}
	$}
	
	\hspace*{2.5cm}
	$\mapsto\quad$
	\scalebox{.85}{$
	\infer=[\weak]{\vec\Gamma\{\vec{\Delta},\vec{\Delta}',\{\vec{\Delta}'',\vec{\Delta}'''\}_+\}_\times \seq R}{
	\infer[\cut]{\vec\Gamma\{\vec{\Delta},\vec{\Delta}',\vec{\Delta}''\}_\times \seq R}{
		\vec{\Delta} \seq \foc{P} 
		&
		\infer[\cut]{\vec\Gamma\{\vec{\Delta},\vec{\Delta}'',P\}_\times \seq R}{
			\{\vec{\Delta}'', P\}_\times \seq  N
			&
			\vec\Gamma\{\vec{\Delta}', \foc{N}\}_\times \seq R
			}
		}
	}
	$}\\

We denote by a double-line the fact that we do not actually use a weakening, but only the fact that it is admissible in $\fbi$ by construction (Lemma~\ref{lem:weakelim}).
\qed
%	\begin{minipage}{\linewidth}
%		\begin{minipage}{0.7\linewidth}
%			\begin{scprooftree}{0.8}
%				\AxiomC{$\{\vec{\Delta}'', P\}_\times \seq  N$ }
%				\RightLabel{$\wand R$}
%				\UnaryInfC{$\vec{\Delta}'' \seq P\wand N$ }
%				\AxiomC{$\vec{\Delta}' \seq \foc{P} $}
%				\AxiomC{$\Gamma\{\vec{\Delta}, \foc{N}\}_\times \seq M $}
%				\RightLabel{$\wand L$}
%				\BinaryInfC{$\Gamma\{\vec{\Delta},\vec{\Delta}',\foc{P\wand N}\}_\times \seq M$}
%				\RightLabel{$\cut$}
%				\BinaryInfC{$\Gamma\{\vec{\Delta},\vec{\Delta}',\vec{\Delta}''\}_\times \seq M$}
%			\end{scprooftree}
%		\end{minipage} 
%		\begin{minipage}{0.25\linewidth}
%			\centering
%			\scalebox{1.5}{$\mapsto$}
%		\end{minipage}
%		\begin{flushright}
%			\begin{minipage}{0.7\linewidth}
%				\begin{scprooftree}{0.8}
%					\AxiomC{$\vec{\Delta}' \seq \foc{P} $}
%					\AxiomC{$\{\vec{\Delta}'', P\}_\times \seq  N$}
%					\AxiomC{$\Gamma\{\vec{\Delta}, \foc{N}\}_\times \seq M $}
%					\RightLabel{$\cut$}
%					\BinaryInfC{$\Gamma\{\vec{\Delta},P,\vec{\Delta}''\}_\times \seq M$}
%					\RightLabel{$\cut$}
%					\BinaryInfC{$\Gamma\{\vec{\Delta},\vec{\Delta}',\vec{\Delta}''\}_\times \seq M$}
%				\end{scprooftree}
%			\end{minipage}
%		\end{flushright}
%	\end{minipage}
\end{proof}

\begin{lemma}[Bad Cuts Elimination]\label{lem:badcutelim}
	Let $\D$ be a $\fbi+\cut$ proof of $S$ that contains only one cut which is bad, then there is a $\fbi+\cut$ proof $\D'$ of $S$ such that $\D' \prec \D$.
\end{lemma}
\begin{proof}
	 Without loss of generality suppose the cut is the last inference in the proof, then it may be replaced by other cuts whose $\cut$-level or $\cut$-duplicity is smaller, but with same $\cut$-complexity.
	
	First we consider bad cuts when $\mathcal{L}$ and $\mathcal{R}$ are both axioms. There are no Type 1 bad cuts on axioms as the formula is always principal, meanwhile the Type $2$ bad cuts can trivially be permuted upwards or ignored; for example,\\
	
	\scalebox{.85}{
	$
	\infer[\cut]{\vec\Gamma\{\vec\Delta,\vec\Delta',\{\vec\Delta'', \vec\Delta''', A_+,\foc{P\wand N}\}_+\}_\times \seq R}{
		\infer[\rn{Ax^+}]{\{\vec\Delta''',A_+\}_+ \seq \foc{A_+}}{}
		&
		\infer[\lrn\wand]{\vec\Gamma\{\vec\Delta,\vec\Delta',\{\vec\Delta'', A_+,\foc{P\wand N}\}_+\}_\times \seq R}{
			\vec\Delta \seq \foc P
			&
			\vec\Gamma\{\vec\Delta', \foc N\}_\times \seq R
			}
		}
	$
	}\\
	
	\hspace*{2.5cm}
	$\mapsto\quad$
	\scalebox{.85}{
	$
	\infer=[\weak]{\vec\Gamma\{\vec\Delta,\vec\Delta',\{\vec\Delta'', \vec\Delta''', A_+,\foc{P\wand N}\}_+\}_\times \seq R}{
	\infer[\lrn\wand]{\vec\Gamma\{\vec\Delta,\vec\Delta',\{\vec\Delta'', A_+,\foc{P\wand N}\}_+\}_\times \seq R}{
		\vec\Delta \seq \foc P
		&
		\vec\Gamma\{\vec\Delta', \foc N\}_\times \seq R
	}
}
	$
	}

	Here again we are using an appropriate version of Lemma~\ref{lem:weakelim}.
%	\begin{minipage}{\linewidth}
%		\begin{scprooftree}{0.8}
%			\AxiomC{$\{\Delta''',A_+\}_+ \seq \foc{A_+}$}
%			\AxiomC{$\Delta \seq \phi$}
%			\AxiomC{$\{\Delta',\psi\}_\times \seq \chi$}
%			\BinaryInfC{$\Gamma\{\Delta,\Delta',\{\Delta''\{A_+\},\foc{\phi \wand \psi}\}_+\}_\times \seq \chi$}
%			\BinaryInfC{$\Gamma\{\Delta,\Delta',\{\Delta''\{\{\Delta''',A_+\}\},\foc{\phi \wand \psi}\}_+\}_\times \seq \chi$}
%		\end{scprooftree}
%	\end{minipage} \\
%	
%	\vspace{0.2cm}
%	\begin{minipage}{0.3\linewidth}
%		\centering
%		\scalebox{1.5}{$\mapsto$}
%	\end{minipage}
%	\begin{minipage}{0.65\linewidth}
%		\begin{scprooftree}{0.8}
%			\AxiomC{$\Delta \seq \phi$}
%			\AxiomC{$\{\Delta',\psi\}_\times \seq \chi$}
%			\BinaryInfC{$\Gamma\{\Delta,\Delta',\{\Delta''\{\{\Delta''',A_+\}\},\foc{\phi \wand \psi}\}_+\}_\times \seq \chi$}
%		\end{scprooftree}
%	\end{minipage} \\ \vspace{0.1cm}
	
	For the remaining cases the cuts are commutative in the sense that they may be permuted upward thereby reducing the $\cut$-level. An example is given below.\\
	
	\scalebox{.85}{
	$
	\infer[\cut]{\vec\Gamma\{\vec\Delta\{\foc{N_1 \land^- N_2}\}\} \seq R}{
		\infer[\text{$\lrn[1]{\land^-}$}]{\vec\Delta\{\foc{N_1 \land^- N_2}\} \seq M}{
			\vec\Delta\{\foc{N_1}\} \seq M
			}
		&
		\vec\Gamma\{  M\} \seq R
		}
	$
	}
%
%	\hspace*{3cm}
	$\mapsto\quad$
	\scalebox{.85}{
	$
	\infer[\text{$\lrn[1]{\land^-}$}]{\Gamma\{\Delta\{\foc{N_1 \land^- N_2}\}\} \seq R}{
		\infer[\cut]{\Gamma\{\Delta\{\foc{N_1}\}\} \seq R}{
			\Delta\{\foc{N_1}\} \seq M
			&
			\Gamma\{  M\} \seq R
			}
		}
	$
	}

%	
%	\begin{minipage}{\linewidth}
%		\begin{minipage}{0.55\linewidth}
%			\begin{scprooftree}{0.8}
%				\AxiomC{$\Delta\{\foc{N_i}\} \seq N$}
%				\RightLabel{$\land^- L$}
%				\UnaryInfC{$\Delta\{\foc{N_1 \land^- N_2}\} \seq N$}
%				\AxiomC{$\Gamma\{  N \} \seq M$}
%				\RightLabel{$\cut$}
%				\BinaryInfC{$\Gamma\{\Delta\{\foc{N_1 \land^- N_2}\}\} \seq M$}
%			\end{scprooftree}
%		\end{minipage}
%		\begin{minipage}{0.05\linewidth}
%			\centering
%			$\mapsto$
%		\end{minipage}
%		\begin{minipage}{0.35\linewidth}
%			\begin{scprooftree}{0.8}
%				\AxiomC{$\Delta\{\foc{N_i}\} \seq N$}
%				\AxiomC{$\Gamma\{  N \} \seq M$}
%				\BinaryInfC{$\Gamma\{\Delta\{\foc{N_i}\} \seq M$}
%				\UnaryInfC{$\Gamma\{\Delta\{\foc{N_1 \land^- N_2}\} \seq M$}
%			\end{scprooftree}
%		\end{minipage}
%	\end{minipage} \\
%	\vspace{0.2cm}

	The exceptional case is the interaction with contraction where the cut is replaced by cuts of possibly equal $\cut$-level, but $\cut$-duplicity decreases.\\
	
	\scalebox{.85}{
	$
	\infer[\cut]{\vec\Gamma\{\vec{\Delta}\{\vec\Delta'\}\} \seq R}{
		\vec\Delta' \seq \foc{L}
		&
		\infer[\cont]{\vec\Gamma\{\vec{\Delta}\{L\}\} \seq R}{
			\vec\Gamma\{\{\vec{\Delta}\{L\}, \vec{\Delta}\{L\}\}_+\} \seq R 
			}
		}
	$
	}

	\hspace*{2.5cm}
	$\quad\mapsto\quad$
	\scalebox{.85}{
	$
	\infer[\cont]{\vec\Gamma\{\vec{\Delta}\{\vec\Delta'\}\} \seq R}{
		\infer[\cut]{\vec\Gamma\{\{\vec{\Delta}\{\vec\Delta'\}, \vec{\Delta}\{\vec\Delta'\}\}_+\} \seq R }{
			\vec\Delta' \seq \foc L
			&
			\infer[\cut]{\vec\Gamma\{\{\vec{\Delta}\{\vec\Delta'\}, \vec{\Delta}\{L\}\}_+\} \seq R}{
				\vec\Delta' \seq \foc L
				&
				\vec\Gamma\{\{\vec{\Delta}\{L\}, \vec{\Delta}\{L\}\}_+\} \seq R
				}
			}
		}
	$
	}
%
%	\begin{minipage}{0.85\linewidth}
%		\begin{scprooftree}{0.8}
%			\AxiomC{$\Delta' \seq N$}
%				\AxiomC{$\Gamma\{\{\vec{\Delta}\{\downshift N\}, \vec{\Delta}\{\downshift N\}\}_+\} \seq \chi $}
%				\UnaryInfC{$\Gamma\{\{\vec{\Delta}\{\downshift N\}\}_+\} \seq \chi$}
%			\BinaryInfC{$\Gamma\{\vec{\Delta}\{\Delta'\}\} \seq \chi $}
%		\end{scprooftree}
%	\end{minipage}
%\begin{minipage}{0.1\linewidth}
%	\centering
%	$\mapsto$
%\end{minipage}
%	\begin{minipage}{\linewidth}
%	\begin{scprooftree}{0.8}
%		\AxiomC{$\Delta' \seq N$}
%		\AxiomC{$\Delta' \seq N$}
%		\AxiomC{$\Gamma\{\{\vec{\Delta}\{\downshift N\}, \vec{\Delta}\{\downshift N\}\}_+\} \seq \chi $}
%		\BinaryInfC{$\Gamma\{\{\vec{\Delta}\{\Delta'\}, \vec{\Delta}\{\downshift N\}\}_+\} \seq \chi $}
%		\BinaryInfC{$\Gamma\{\{\vec{\Delta}\{\Delta'\}, \vec{\Delta}\{\Delta'\}\}_+\} \seq \chi $}
%		\UnaryInfC{$\Gamma\{\vec{\Delta}\{\Delta'\}\} \seq \chi $}
%	\end{scprooftree}
%\end{minipage}
%
\qed\end{proof}

\begin{theorem}[Cut-elimination in $\fbi$]\label{lem:fbicutadmi}
	Let $\Gamma$ be a positive nest and $N$ a negative formula. Then, 
	$\vdash_{\fbi} \Gamma \seq N$ if and only if $ \vdash_{\fbi+\cut} \Gamma \seq N$.
\end{theorem}
\begin{proof}
	($\seq$) Trivial as any $\fbi$-proof is a $\fbi+\cut$-proof. ($\Leftarrow$) Let $\D$ be a $\fbi+\cut$-proof of $\Gamma \seq N$, if it has no cuts then it is a $\fbi$-proof so we are done. Otherwise, there is at least one $\cut$, and we proceed by well-founded induction on the ordering of proofs and sub-proofs of $\D$ with respect to $\prec$.
	
% 	\begin{enumerate}
% 		\item[] 
 		\textbf{Base Case.} Assume $\D$ is minimal with respect to $\prec$ with at least one cut; without loss of generality, by Lemma \ref{lem:goodcutelim}, assume the cut is bad. It follows from  Lemma \ref{lem:badcutelim} that there is a proof strictly smaller in $\prec$-ordering, but this proof must be $\cut$-free as $\D$ is minimal.

% 		\item[]	
 		\textbf{Inductive Step.} Let $\D$ be as in the hypothesis, then by Lemma \ref{lem:goodcutelim} there is a proof $\partial$ of $\Gamma \seq N$ containing no good cuts such that $\D' \preceq \D$. Either $\D'$ is $\cut$-free and we are done, or it contains bad cuts. Consider the topmost cut, and denote the sub-proof by $\partial$, it follows from Lemma \ref{lem:badcutelim} that there is a proof $\partial'$ of the same sequent such that $\partial' \prec \partial$. Hence, by inductive hypothesis, there is a $\cut$-free proof the sequent and replacing $\partial$ by this proof in $\D$ gives a proof of $\Gamma \seq \phi$ strictly smaller in $\prec$-ordering, thus by inductive hypothesis there is a $\cut$-free proof as required.
% 	\end{enumerate}
\qed \end{proof}

\subsection{Completeness of $\fbi$} \label{sec:foc_comp}
The completeness theorem of the focused system, the operational semantics, is with respect to an interpretation (i.e. a polarisation). Indeed, any polarisation may be considered; for example, both $(\downshift A^-*B^+) \land^+ \downshift A^-$ and $\downshift (A^+*\downshift B^-) \land^+ A^+$ are correct polarised versions of the formulas $(A*B) \land A$. Taking arbitrary $\phi$ the process is as follows: first, fix a polarised syntax (i.e. a partition of the propositional letters into positive and negative sets), then assign a polarity to $\phi$ with the following steps:
\begin{itemize}
	\item If $\phi$ is a propositional atom, it must be polarised by default; 
	\item If $\phi = \top$, then \emph{choose} polarisation $\top^+$ or $\top^-$;
	\item If $\phi = \psi_1 \land \psi_2$, first polarise $\psi_1$ and $\psi_2$, then \emph{choose} an additive conjunction and combine accordingly, using shifts to ensure the formula is well-formed;
	\item If $\phi = \psi_1 \circ \psi_2$ where $\circ \in\{*, \wand, \to, \lor\}$, then polarise $\psi_1$ and $\psi_2$ and combine with $\circ$ accordingly, using shifts where necessary.
\end{itemize}
\begin{example}
	Suppose $A$ is negative and $B$ is positive, then $(A*B) \land A$ may be polarised by choosing the additive conjunction to be positive resulting in $(\downshift A*B) \land^+ \downshift A$ (when $\downshift (A*\downshift B) \land^+ A)$ would not be well-formed). Choosing to shift one can ascribe a negative polarisation  $\upshift ((\downshift A*B) \land^+ \downshift A)$.
\end{example}
The above generates the set of all such polarised formulas when all possible choices are explored. The free assignment of polarity to formulas means several distinct focusing procedures are captured by the completeness theorem. 

\begin{lemma}[Completeness of $\fbi+\cut$] \label{lem:compFBIcut}
		For any unfocused sequent $\Gamma \seq N$, if $\vdash_{\nestify \lbi}\floor{\Gamma \seq N}$ then $\vdash_{\fbi+\cut} \Gamma \seq N$.
\end{lemma}
\begin{proof}
	We show that every rule in $\nestify\lbi$ is derivable in $\fbi+\cut$, consequently every proof in $\nestify\lbi$ may be simulated; hence, every provable sequent has a focused proof. For unfocused rules $\rrn\to, \rrn\wand, \rrn{\land^-}, \lrn{\land^+}, \lrn\lor, \lrn\ast, \lrn\bot, \rrn{\top^-}, \lrn{\top^+}, \lrn\i$, this is immediate; as well as for $\rn{Ax}$ and $\cont$. Below we give an example on how to simulate a focused rule.
	
	Where it does not matter (e.g. in the case of inactive nests), we do not distinguish the polarised and unpolarised versions; each of the simulations can be closed thanks to the presence of the $\rn P$ and $\rn N$ rules in $\fbi$.\\
	
	\noindent\scalebox{.9}{	
	$
	\infer[\rrn\ast]{\{\{\Gamma,\Delta\}_\times,\Delta'\}_+ \seq \phi * \psi}{\Gamma \seq \phi & \Delta \seq \psi}
	$
	}
	in $\nestify\lbi$ is simulated in $\fbi+\cut$ by
	
	\scalebox{.9}{
	$
	\infer[\cut]{\{\{ \Gamma, \Delta \}_\times,\Delta'\}_+ \seq \upshift(\phi^+*\psi^+)}{
		\infer[\rrn\ast]{\{ \Gamma,\Delta \}_\times \seq \foc{\downshift\upshift\phi^+*\downshift\upshift\psi^+} }{
			\infer[\rrn\downshift]{\Gamma \seq \foc{\downshift\upshift\phi^+}}{
				\Gamma \seq \upshift\phi^+
				}
			&
			\infer[\rrn\downshift]{\Delta \seq \foc{\downshift\upshift\psi^+}}{\Delta \seq \upshift\psi^+}
			}
		&
		\infer[\lrn\ast]{\{ \downshift\upshift\phi^+*\downshift\upshift\psi^+,\Delta'\}_+ \seq \upshift(\phi^+*\psi^+)}{
			\infer[\rrn\upshift]{\{\{ \downshift\upshift\phi^+,\downshift\upshift\psi^+ \}_\times,\Delta'\}_+ \seq \upshift(\phi^+*\psi^+)}{
				\infer[\rrn\ast]{ \{\{ \downshift\upshift\phi^+,\downshift\upshift\psi^+ \}_\times,\Delta'\}_+ \seq \foc{\phi^+*\psi^+}}{
					\infer[\rn P]{\downshift\upshift\phi^+ \seq \foc{\phi^+}}{}
					&
					\infer[\rn P]{\downshift\upshift\psi^+ \seq \foc{\psi^+}}{}
					}
				}
			}
		}
	$
	}
\qed \end{proof}
\begin{theorem}[Completeness of $\fbi$]
	For any unfocused $\Gamma \seq N$, if $\vdash_{\nestify\lbi} \floor{\Gamma \seq N}$ then $\vdash_{\fbi} \Gamma \seq N$.
\end{theorem}
\begin{proof}
	It follows from Lemma \ref{lem:compFBIcut} that there is a proof of $\Gamma \seq N$ in $\fbi+\cut$, and then it follows from Lemma \ref{lem:fbicutadmi} that there is a proof of $\Gamma \seq N$ in $\fbi$.
\qed \end{proof}
Given an arbitrary sequent the above theorem guarantees the existence of a focused proof, thus the focusing principle holds for $\eta \lbi$ and therefore for $\lbi$.

%% file: conclusion.tex
By proving the completeness of a focused sequent calculus for the logic of Bunched Implications, we have demonstrated that it satisfies the focusing principle; that is, any polarisation of a BI-provable sequent can be proved following a focused search procedure. This required a careful analysis of how to restrict the usage of structural rules. In particular, we had to fully develop the congruence-invariant representation of bunches as nested multisets (originally proposed in \cite{Donnelly05}) to treat the exchange rule within bunched structures.

Proof-theoretically the completeness of the focused systems suggests a syntactic orderliness of $\lbi$, though the $\rn P$ and $\rn N$ rules leave something to be desired. Computationally, these axioms are unproblematic as during search it makes sense to terminate a branch as soon as possible; however, unless they may be eliminated it means that the focusing principle holds in BI only up to a point. 
%Restricting attention to the reductive logic of BI, it may be that focused proof-search can be used to implement efficient and effective proof-search procedures.  
In related works (c.f. \cite{Chaudhuri16a}) the analogous problem is overcome by first considering a \emph{weak} focused system; that is, one where the structural rules are not controlled and unfocused rules may be performed inside focused phases if desired. Completeness of (strong) focusing is achieved by appealing to a \emph{synthetic} system. It seems reasonable to suppose the same can be done for BI, resulting in a more proof-theoretically satisfactory focused calculus, exploring this possibility is a natural extension of the work on $\fbi$.

The methodology employed for proving the focusing principle can be interpreted as soundness and completeness of an operational semantics for goal-directed search. The robustness of this technique is demonstrated by its efficacy in modal \cite{Chaudhuri16a,Chaudhuri16b} and substructural logics \cite{Lincoln92}, including now bunched ones. Although BI may be the most employed bunched logic, there are a number of others, such as the family of relevant logics \cite{Read88}, and the family of bunched logics \cite{Docherty19}, for which the focusing principle should be studied. However, without the presence of a $\cut$-free sequent calculus goal-directed search becomes unclear, and currently such calculi do not exist for the two main variants of BI: Boolean BI \cite{Pym02} and Classical BI \cite{Brotherston10}. On the other hand, large families of bunched and substructural logics have been given hypersequent calculi \cite{Ciabattoni12,Ciabattoni17}. Effective proof-search procedures have been established for the hypersequent calculi in the substructural case \cite{Ramanayake20}, but not the bunched one, and focused proof-search for neither. There is a technical challenge in focusing these systems as one must not only decide which formula to reduce, but also which sequent.  
%\todo{comment on the nested sequent system for BBI?}

In the future it will be especially interesting to see how focused search, when combined with the expressiveness of BI, increases its modelling capabilities. Indeed, the dynamics of proof-search can be used to represent models of computation within (propositional) logics; for example, the undecidability of Linear Logic involves simulating two-counter machines \cite{Lincoln92}. One particularly interesting direction is to see how focused proof-search in BI may prove valuable within the context of Separation Logic. Focused systems in particular have been used to emulate proofs for other logics \cite{Marin16}; and to give structural operational semantics for systems used in industry, such as algorithms for solving constraint satisfaction problems \cite{Farooque13}. A more immediate possibility though is the formulation of a theorem prover; we leave providing specific implementation or benchmarks to future research.

%% file: document.bbl
\begin{thebibliography}{10}
\providecommand{\url}[1]{\texttt{#1}}
\providecommand{\urlprefix}{URL }
\providecommand{\doi}[1]{https://doi.org/#1}

\bibitem{Andreoli92}
Andreoli, J.: Logic programming with focusing proofs in linear logic. Journal
  of Logic and Computation  \textbf{2},  297--347 (1992)

\bibitem{Armelin02}
Armelin, P.: Bunched Logic Programming. Ph.D. thesis, Queen Mary College,
  University of London (2002)

\bibitem{Brotherston10a}
Brotherston, J.: A unified display proof theory for bunched logic. Electronic
  Notes in Theoretical Computer Science  \textbf{265},  197 -- 211 (2010),
  proceedings of the 26th Conference on the Mathematical Foundations of
  Programming Semantics (MFPS 2010)

\bibitem{Brotherston10}
Brotherston, J., Calcagno, C.: Classical {BI}: Its semantics and proof theory.
  Logical Methods in Computer Science  \textbf{6} (05 2010).
  \doi{10.2168/LMCS-6(3:3)2010}

\bibitem{Chaudhuri16b}
Chaudhuri, K., Marin, S., Stra{\ss}burger, L.: Focused and synthetic nested
  sequents. In: Foundations of Software Science and Computation. pp. 390--407
  (04 2016)

\bibitem{Chaudhuri16a}
Chaudhuri, K., Marin, S., Stra{\ss}burger, L.: Modular focused proof systems
  for intuitionistic modal logics. In: FSCD 2016 - 1st International Conference
  on Formal Structures for Computation and Deduction (2016)

\bibitem{Chaudhuri06}
Chaudhuri, K., Pfenning, F., Price, G.: A logical characterization of forward
  and backward chaining in the inverse method. In: Furbach, U., Shankar, N.
  (eds.) Automated Reasoning. pp. 97--111. Springer Berlin Heidelberg, Berlin,
  Heidelberg (2006)

\bibitem{Ciabattoni12}
Ciabattoni, A., Galatos, N., Terui, K.: Algebraic proof theory for
  substructural logics: Cut-elimination and completions. Annals of Pure and
  Applied Logic  \textbf{163}(3),  266–290 (Mar 2012).
  \doi{10.1016/j.apal.2011.09.003},
  \url{http://dx.doi.org/10.1016/j.apal.2011.09.003}

\bibitem{Ciabattoni17}
Ciabattoni, A., Ramanayake, R.: Bunched hypersequent calculi for distributive
  substructural logics. In: Eiter, T., Sands, D. (eds.) LPAR-21. 21st
  International Conference on Logic for Programming, Artificial Intelligence
  and Reasoning. EPiC Series in Computing, vol.~46, pp. 417--434. EasyChair
  (2017). \doi{10.29007/ngp3},
  \url{https://easychair.org/publications/paper/sr2D}

\bibitem{Dershowitz79}
Dershowitz, N., Manna, Z.: Proving termination with multiset orderings. In:
  Maurer, H.A. (ed.) Automata, Languages and Programming. pp. 188--202.
  Springer Berlin Heidelberg, Berlin, Heidelberg (1979)

\bibitem{Docherty19}
Docherty, S.: Bunched Logics: A Uniform Approach. Ph.D. thesis, University
  College London (2019)

\bibitem{Donnelly05}
Donnelly, K., Gibson, T., Krishnaswami, N., Magill, S., Park, S.: The inverse
  method for the logic of bunched implications. In: Baader, F., Voronkov, A.
  (eds.) Logic for Programming, Artificial Intelligence, and Reasoning. pp.
  466--480. Springer Berlin Heidelberg, Berlin, Heidelberg (2005)

\bibitem{Dyckhoff06}
Dyckhoff, R., Lengrand, S.: {LJQ}: A strongly focused calculus for
  intuitionistic logic. In: Beckmann, A., Berger, U., L{\"o}we, B., Tucker,
  J.V. (eds.) Logical Approaches to Computational Barriers. pp. 173--185.
  Springer Berlin Heidelberg, Berlin, Heidelberg (2006)

\bibitem{Farooque13}
Farooque, M., Graham-Lengrand, S., Mahboubi, A.: A bisimulation between
  {DPLL(T)} and a proof-search strategy for the focused sequent calculus. p.
  3–14. LFMTP '13, Association for Computing Machinery, New York, NY, USA
  (2013). \doi{10.1145/2503887.2503892}

\bibitem{Gabbay1998}
Gabbay, D.: Fibring Logics. Oxford Logic Guides, Clarendon Press (1998),
  \url{https://books.google.co.uk/books?id=mpA1uUV-uYsC}

\bibitem{Galmiche2003}
Galmiche, D., M{\'e}ry, D.: Semantic labelled tableaux for propositional {BI}.
  J. Log. Comput.  \textbf{13},  707--753 (2003)

\bibitem{Galmiche05}
Galmiche, D., M\'ery, D., Pym, D.: The semantics of {BI }and resource tableaux.
  Mathematical Structures in Computer Science  \textbf{15}(6),  1033--1088
  (2005)

\bibitem{Galmiche2019}
Galmiche, D., Marti, M., M{\'e}ry, D.: Relating labelled and label-free bunched
  calculi in {BI} logic. In: International Conference on Automated Reasoning
  with Analytic Tableaux and Related Methods. pp. 130--146. Springer (2019)

\bibitem{Galmiche2001}
Galmiche, D., M{\'e}ry, D.: {Proof-search and countermodel generation in
  propositional {BI} Logic - extended abstract -}. In: N.~Kobayashi, B.P. (ed.)
  {4th International Symposium on Theoretical Aspects of Computer Software -
  TACS 2001}. Lecture Notes in Computer Science, vol.~2215, pp. 263--282.
  {Springer}, Sendai, Japan (2001)

\bibitem{Galmiche2002}
Galmiche, D., M{\'e}ry, D.: {Connection-based proof search in propositional
  {BI} logic}. In: Voronkov, A. (ed.) {18th International Conference on
  Automated Deduction - CADE-18}. Lecture Notes in Computer Science, vol.~2392,
  pp. 111--128. {Springer Verlag}, Copenhagen/Denmark (2002)

\bibitem{Gentzen1969}
Gentzen, G.: The Collected Papers of Gerhard Gentzen. Amsterdam: North-Holland
  Pub. Co. (1969)

\bibitem{Girard87}
Girard, J.Y.: Linear logic. Theoretical Computer Science  \textbf{50}(1),  1 --
  101 (1987)

\bibitem{Ishtiaq2011}
Ishtiaq, S., O'Hearn, P.W.: {BI} as an assertion language for mutable data
  structures. SIGPLAN Not.  \textbf{46}(4),  84–96 (May 2011).
  \doi{10.1145/1988042.1988050}, \url{https://doi.org/10.1145/1988042.1988050}

\bibitem{Laurent04}
Laurent, O.: A proof of the focalization property of linear logic (04 2004),
  https://perso.ens-lyon.fr/olivier.laurent/llfoc.pdf

\bibitem{Liang09}
Liang, C., Miller, D.: Focusing and polarization in linear, intuitionistic, and
  classical logics. Journal of Theoretical Computer Science.  \textbf{410}(46),
   4747–4768 (Nov 2009)

\bibitem{Lincoln92}
Lincoln, P., Mitchell, J., Scedrov, A., Shankar, N.: Decision problems for
  propositional linear logic. Annals of Pure and Applied Logic  \textbf{56}(1),
   239 -- 311 (1992)

\bibitem{Marin16}
Marin, S., Miller, D., Volpe, M.: {A focused framework for emulating modal
  proof systems} (11),  469--488 (2016)

\bibitem{Mclaughlin08}
McLaughlin, S., Pfenning, F.: Imogen: Focusing the polarized focused inverse
  method for intuitionistic propositional logic. In: Cervesato, I., Veith, H.,
  Voronkov, A. (eds.) 15th International Conference on Logic, Programming,
  Artificial Intelligence and Reasoning (LPAR). vol.~5330, pp. 174--181 (Nov
  2008)

\bibitem{Miller91}
Miller, D., Nadathur, G., Pfenning, F., Scedrov, A.: Uniform proofs as a
  foundation for logic programming. Annals of Pure and Applied Logic
  \textbf{51}(1),  125 -- 157 (1991)

\bibitem{Miller13}
Miller, D., Pimentel, E.: A formal framework for specifying sequent calculus
  proof systems. Theoretical Compututer Science  \textbf{474},  98–116 (Feb
  2013)

\bibitem{Hearn99}
O'Hearn, P., Pym, D.: The logic of bunched implications. The Bulletin of
  Symbolic Logic  \textbf{5}(2),  215--244 (1999)

\bibitem{Pym19}
Pym, D.: Resource semantics: Logic as a modelling technology. ACM SIGLOG News
  \textbf{6}(2),  5–41 (Apr 2019). \doi{10.1145/3326938.3326940},
  \url{https://doi.org/10.1145/3326938.3326940}

\bibitem{Pym02}
Pym, D.J.: The Semantics and Proof Theory of the Logic of Bunched Implications,
  Applied Logic Series, vol.~26. Springer Netherlands, Dordrecht (2002)

\bibitem{Pym2004resource}
Pym, D.J., O'Hearn, P.W., Yang, H.: {Possible Worlds and Resources: the
  Semantics of BI}. Theoretical Computer Science  \textbf{315}(1),  257 -- 305
  (2004). \doi{https://doi.org/10.1016/j.tcs.2003.11.020},
  \url{http://www.sciencedirect.com/science/article/pii/S0304397503006248}

\bibitem{Ramanayake20}
Ramanayake, R.: Extended {K}ripke lemma and decidability for hypersequent
  substructural logics. p. 795–806. LICS '20, Association for Computing
  Machinery, New York, NY, USA (2020). \doi{10.1145/3373718.3394802}

\bibitem{Read88}
Read, S.: Relevant Logic: A Philosophical Examination of Inference. B.
  Blackwell (1988)

\bibitem{Reynolds02}
{Reynolds}, J.C.: Separation logic: a logic for shared mutable data structures.
  In: Proceedings 17th Annual IEEE Symposium on Logic in Computer Science. pp.
  55--74 (2002)

\end{thebibliography}
